\documentclass[sn-mathphys-num]{sn-jnl}
\usepackage{booktabs,longtable,array,multirow,float}
\usepackage[table]{xcolor}
\usepackage{graphicx,amsmath,amssymb,amsfonts,amsthm,mathrsfs}
\usepackage[title]{appendix}
\usepackage{textcomp,manyfoot,algorithm,algorithmicx,algpseudocode}
\usepackage{dsfont,geometry,natbib,hyperref,bm}
\theoremstyle{thmstyleone}
\newtheorem{theorem}{Theorem}[section]
\newtheorem{proposition}[theorem]{Proposition}
\newtheorem{corollary}[theorem]{Corollary}

\theoremstyle{thmstyletwo}

\newtheorem{remark}[theorem]{Remark}

\theoremstyle{thmstylethree}
\newtheorem{definition}[theorem]{Definition}

\usepackage{hyperref}

\hypersetup{
	pdftitle={Quantifying the 2027 Solvency II Risk Margin Reform},
	pdfauthor={Khalil Said and Fatima Zahrae Chaayra},
	pdfsubject={Actuarial and quantitative analysis of the 2027 Solvency II Risk Margin reform},
	pdfkeywords={Solvency II, Risk Margin, Cost of Capital, Solvency Capital Requirement, Technical Provisions, Insurance Valuation, Market-Consistent Valuation, Interest Rates},
	pdfdisplaydoctitle=true
}

\begin{document}

\title[Quantifying the 2027 Solvency II Risk Margin Reform]{Quantifying the 2027 Solvency II Risk Margin Reform}

\author*[1,2]{\fnm{Khalil} \sur{Said}}\email{khalil.said@univ-pau.fr}
\author[2]{\fnm{Fatima Zahrae} \sur{Chaayra}}\email{fzchaayra@insea.ac.ma}
\affil[1]{
	\orgname{Université de Pau et des Pays de l’Adour, E2S UPPA, TREE, Pau},
	\orgaddress{\country{France}}} 
\affil[2]{
	\orgname{National Institute of Statistics and Applied Economics, Rabat}, \orgaddress{\country{Morocco}}}

\abstract{%
	The 2027 Solvency II reform recalibrates the Risk Margin by reducing the prescribed cost-of-capital rate from 6\% to 4.75\% and introducing a time-dependent attenuation of future Solvency Capital Requirements. This paper develops an analytical and numerical framework for characterizing the effect of the final regulatory calibration. By normalizing discounted projected capital requirements into a probability distribution over run-off time, we obtain an exact representation of the ratio between the revised and previous Risk Margins. The framework yields sharp universal and horizon-specific bounds, characterizes the effect of later capital timing through stochastic dominance, and shows that mean run-off time alone does not determine the reform effect when temporal dispersion varies. Additional bounds are derived conditional on mean run-off time and horizon. Conditional on a given projected SCR path, the mechanical reduction lies between 20.83\% and 60.42\%. For proportional Best Estimate projections, an exact covariance decomposition identifies how departures from proportionality affect the relative reform ratio. Holding the projected SCR path fixed, the revised formula has a lower direct interest-rate semi-elasticity than the previous calibration. A reduced-form stochastic extension further quantifies convexity effects arising from uncertainty in capital persistence. Numerical applications reconstruct Risk Margin calculations from published actuarial run-off profiles and complement them with controlled long-horizon, interest-rate, and persistence experiments. The results show that the reform effect is governed by the temporal structure of future capital and provide tractable tools for assessing projected capital profiles, Risk Margin simplifications, and direct discount-rate sensitivity.
}


\maketitle
\section*{Introduction}

On 30 January 2027, the Solvency II Risk Margin will change from a cost-of-capital calculation with a constant regulatory loading into one that explicitly differentiates future capital according to its position in the run-off. The amendment is compact in regulatory form, but its effect is not uniform across insurance portfolios. A unit of capital required in the first years of the run-off will no longer receive the same regulatory treatment as an otherwise identical unit required much later. The reform therefore changes not only the level of the Risk Margin, but also the way in which future capital contributes to it over time.

Directive (EU) 2025/2 establishes the revised legislative framework, including a prescribed cost-of-capital rate of 4.75\% from 30 January 2027 and a time-dependent adjustment to future capital requirements~\cite{europeanunion2025solvency}. Commission Delegated Regulation (EU) 2026/269 specifies the final numerical calibration of that adjustment \cite{europeancommission2026delegated}. If \(SCR_t\) denotes the projected Solvency Capital Requirement of the reference undertaking and \(D_t\) the corresponding risk-free discount factor, the previous and revised calculations can be written as

\[
\begin{aligned}
	RM_{\mathrm{old}}
	&=
	0.06\sum_{t\geq0}SCR_tD_t,\\
	RM_{\mathrm{new}}
	&=
	0.0475\sum_{t\geq0}
	\max\!\left(0.96^t,0.50\right)
	SCR_tD_t.
\end{aligned}
\]

The reduction in the cost-of-capital rate produces the same proportional effect on every future capital contribution. The time-dependent adjustment does not. Its impact depends on the timing of the discounted SCR sequence and becomes progressively stronger along the run-off until the regulatory floor is reached. The revised formula consequently introduces a run-off timing effect that cannot be summarized by the change in the cost-of-capital rate alone.

This recalibration represents a substantive change in the design of the Solvency II Risk Margin. For non-replicable insurance obligations, Solvency II technical provisions combine a Best Estimate with a Risk Margin representing the cost associated with the capital required by a reference undertaking to support the obligations over their remaining lifetime. The cost-of-capital construction translates this transfer-value principle into a sequence of future regulatory capital costs. The resulting valuation problem lies at the intersection of actuarial reserving, capital measurement and incomplete-market valuation.

The theoretical foundations of this construction have been examined well beyond its regulatory implementation. Salzmann and W\"uthrich (2010)\cite{salzmann2010cost} developed a cost-of-capital margin for general insurance liability run-off, while M\"ohr (2011)\cite{mohr2011market} developed a multi-period market-consistent replication framework with recursive backward valuation and showed that a Best Estimate plus Risk Margin representation can be constructed as a simplification providing an upper bound under suitable conditions. Floreani (2011)\cite{floreani2011risk} examined the economic assumptions underlying the prescribed methodology and questioned the theoretical justification for a fixed unitary cost-of-capital rate. W\"uthrich et al. (2011)\cite{wuthrich2011risk} and Tsanakas et al. (2013)\cite{tsanakas2013market} further studied non-life run-off risk and market value margins from risk-measure and hedging perspectives.

A broader literature subsequently embedded these questions in market-consistent and time-consistent insurance valuation. Pelsser and Stadje (2014)\cite{pelsser2014time} introduced and characterized two-step market evaluations combining actuarial principles with market consistency and established their relation to dynamic time consistency. Engsner et al. (2017)\cite{engsner2017insurance} formulated a computable multi-period cost-of-capital approach incorporating capital requirements, limited liability and the preferences of capital providers. Dhaene et al. (2017)\cite{dhaene2017fair} characterized fair insurance valuation through the interaction between market-consistent hedging and actuarial valuation, while Barigou et al. (2019)\cite{barigou2019fair} extended this perspective to a dynamic multi-period setting. Barigou et al. (2022)\cite{barigou2022insurance} developed a two-step generalized regression approach to insurance valuation, and Gambaro (2025)\cite{gambaro2025capital} formalized the capital-on-capital component associated with uncertainty in future regulatory capital. Together, these contributions place the Risk Margin within a dynamic valuation problem involving unhedgeable risk, future capital requirements and their evolution over time.

The projection of future capital forms a second strand of the literature that is particularly relevant to the 2027 reform. Solvency II combines a one-year solvency perspective with obligations that may remain outstanding for many years. Ohlsson and Lauzeningks (2009)\cite{ohlsson2009oneyear} clarified the one-year non-life insurance risk and showed how reserve duration can support simplified Risk Margin calculations. W\"uthrich et al. (2009)\cite{wuthrich2009uncertainty}, Diers and Linde (2013)\cite{diers2013multiyear}, and Diers et al. (2016)\cite{diers2016addendum} developed the analysis of claims development uncertainty and multi-year reserve risk. Christiansen and Niemeyer (2014)\cite{christiansen2014fundamental} compared alternative mathematical interpretations of the Solvency Capital Requirement and generalized its definition to future points in time, while Bauer et al. (2012)\cite{bauer2012calculation} developed a mathematical framework for SCR calculation and compared numerical implementations based on nested simulations.

These difficulties explain the practical importance of simplified SCR projections. Dreksler et al. (2015)\cite{dreksler2015solvency} provide a detailed actuarial treatment of Solvency II technical provisions and discuss Risk Margin approximations based on proportional run-off and duration-type methods. England et al. (2019)\cite{england2019lifetime} connect the lifetime and one-year views of reserve risk and compare future capital profiles based on the Best Estimate, the standard deviation and the \(99.5\%\) Value-at-Risk. In their illustrative example, the Best Estimate profile is close to the standard-deviation and Value-at-Risk profiles, while the authors emphasize that this proximity is data-dependent. Pelkiewicz et al. (2020)\cite{pelkiewicz2020review} similarly emphasize the practical difficulty of projecting future SCRs and the sensitivity of simplified approaches to assumptions about the evolution of risk.

The calibration of the cost of capital has been questioned separately. Albrecher et al. (2022)\cite{albrecher2022cost} provide an economic foundation for the cost-of-capital rate through the equilibrium between policyholders, shareholders and the regulator, illustrating that the parameter is not economically innocuous. Concerns have also been raised about the magnitude, volatility and interest-rate sensitivity of the Risk Margin, particularly for long-duration liabilities \cite{pelkiewicz2020review}. Waszink (2024)\cite{waszink2024comments} examines alternative interpretations of the cost-of-capital construction and earlier reform proposals, showing that economically defensible specifications can generate substantially different outcomes for long liabilities. Gambaro (2025)\cite{gambaro2025capital} adds a further dynamic dimension through the uncertainty surrounding future regulatory capital itself.

The final European reform addresses several of these concerns. Commission Delegated Regulation (EU) 2026/269 refers to insufficient recognition of the natural decline of some risks over time and to the possibility of double counting, including for lapse and mortality risks \cite{europeancommission2026delegated}. The lower cost-of-capital rate produces an immediate proportional reduction, while the exponential factor further attenuates capital contributions arising later in the run-off. The revised Guidelines published by the European Insurance and Occupational Pensions Authority (EIOPA) in July 2026 describe the lambda factor as accounting for the time dependency of risks and reducing the Risk Margin particularly for long-term liabilities \cite{eiopa2026valuation}. They retain the hierarchy of SCR projection simplifications and require projected SCRs to avoid incorporating a second time the temporal effect already represented by the lambda factor. The accompanying Final Report sets out draft Regulatory Technical Standards intended to align the simplified Risk Margin calculation under Article 58 with the revised framework \cite{eiopa2026riskmarginrts}.

The final calibration raises a quantitative question that is distinct from the earlier debate over the economic design of the Risk Margin. Once the cost-of-capital rate and the lambda sequence are fixed, the effect of the reform remains portfolio dependent because future capital is distributed differently over time. The relevant problem is therefore to determine which features of the SCR run-off govern that effect and how much information about the full capital profile is required to quantify or bound it. The existing literature reviewed above addresses the foundations of cost-of-capital valuation, dynamic liability valuation, future SCR projection, simplification methods and alternative Risk Margin designs. To the best of our knowledge, however, it does not provide an integrated analytical treatment of the final calibration in terms of run-off timing, temporal dispersion, proportional projection error and direct discount-rate sensitivity.

This paper provides such a characterization by normalizing discounted future SCR contributions into an auxiliary probability distribution over run-off time. Under this representation, the ratio between the revised and previous Risk Margins becomes an expectation of the regulatory lambda factor, providing a common structure for the subsequent analysis. The framework yields sharp global and finite-horizon bounds, establishes the effect of later capital timing through stochastic ordering, and separates mean run-off timing from temporal dispersion. Additional sharp bounds are obtained conditional on both the projection horizon and the mean discounted-SCR run-off time.

The framework also gives an exact characterization of proportional Best Estimate projections. Defining future capital intensity as the ratio between the projected SCR and the Best Estimate, the difference between the full reform ratio and EIOPA's Method 2 approximation is determined by a covariance between capital intensity and the lambda sequence. This identifies the conditions under which proportional projection overstates or understates the relative reform effect, while remaining distinct from the approximation error in the absolute Risk Margin level. Holding a common projected SCR path fixed under parallel shifts in continuously compounded discount rates, the lambda weighting also shifts relative weight towards earlier capital contributions. The direct interest-rate semi-elasticity of the revised Risk Margin is consequently lower than under the previous calibration. The difference admits an exact covariance representation, which also characterizes the response of the new-to-old Risk Margin ratio to the imposed rate shift.

The deterministic framework is further extended to uncertainty in capital run-off persistence. A reduced-form stochastic persistence model quantifies the convexity effect generated when an uncertain persistence parameter is replaced by its mean and derives bounds on the associated Jensen gap under the revised calibration. A comparison between structural persistence and independent transitory variation further isolates the effect of uncertainty that persists across the capital run-off.

The numerical analysis combines these results with published actuarial benchmarks and controlled experiments. Capital profiles reported by England et al. (2019)\cite{england2019lifetime} are reconstructed under their original assumptions and evaluated under the 2027 calibration. The published illustrative SCR run-off in Dreksler et al. (2015)\cite{dreksler2015solvency} provides a second benchmark. Geometric capital profiles, parallel interest-rate shifts and stochastic persistence experiments are then used to isolate the comparative-static mechanisms established analytically. The numerical evidence therefore complements the theoretical results without making them dependent on a particular portfolio specification.

The remainder of the paper develops the analytical framework, derives the timing and sharp-bound results, examines Best Estimate approximations and direct interest-rate sensitivity, and extends the analysis to stochastic run-off persistence. The numerical section evaluates these mechanisms across published and controlled capital profiles, before the discussion considers their implications for Risk Margin projection, regulatory simplifications and the interpretation of the 2027 reform.

\section{Analytical framework}
\label{AnalyticalFramework}

The revised Solvency II Risk Margin preserves the cost-of-capital structure of the previous calibration while modifying two components of the calculation. The prescribed cost-of-capital rate falls from \(6\%\) to \(4.75\%\), and projected future Solvency Capital Requirements are multiplied by a time-dependent attenuation factor \cite{europeancommission2026delegated,eiopa2026valuation}. The second modification differentiates capital contributions according to their position in the run-off and therefore depends intrinsically on the temporal profile of future capital.

The analysis isolates the mechanical effect of this revised calibration. Unless stated otherwise, the projected SCR path is held fixed when the previous and revised formulas are compared. This convention defines a controlled comparative-static exercise focused on the direct effect of the regulatory recalibration. Throughout, \(SCR_t\) denotes the projected capital requirement of the reference undertaking before application of the new time-dependent factor. This convention is consistent with EIOPA's requirement that the time dependency already captured by the lambda factor should not be incorporated a second time into projected SCRs \cite{eiopa2026valuation}.

\subsection{Regulatory setting and notation}
\label{RegulatorySetting}

Let \(t=0,1,\ldots\) index future annual projection dates, and let \(SCR_t \geq 0\) denote the projected Solvency Capital Requirement of the reference undertaking at time \(t\). Define

\[
D_t
=
\frac{1}
{\left(1+r_{t+1}\right)^{t+1}},
\qquad t\geq0,
\]

where \(r_{t+1}>-1\) denotes the relevant annual effective risk-free zero-coupon rate for maturity \(t+1\).

Under the previous calibration, the Risk Margin is

\begin{equation}
	RM^{\mathrm{old}}
	=
	c_0
	\sum_{t\geq0}
	SCR_tD_t,
	\qquad
	c_0=0.06.
	\label{RMOld}
\end{equation}

Under the revised calibration applicable from 30 January 2027,

\begin{equation}
	RM^{\mathrm{new}}
	=
	c_1
	\sum_{t\geq0}
	\lambda_t SCR_tD_t,
	\qquad
	c_1=0.0475,
	\label{RMNew}
\end{equation}

where the time-dependent regulatory factor is denoted by

\begin{equation}
	\lambda_t
	=
	\max\left(0.96^t,0.50\right).
	\label{LambdaFactor}
\end{equation}

The notation follows the terminology used by EIOPA in the revised Guidelines on the valuation of technical provisions \cite{eiopa2026valuation}. For integer \(t\),

\[
\lambda_t
=
\begin{cases}
	0.96^t, & 0\leq t\leq16,\\
	0.50, & t\geq17,
\end{cases}
\]

so that the regulatory floor first becomes binding at \(t=17\).

Assume throughout that

\[
0
<
A
:=
\sum_{t\geq0}SCR_tD_t
<
\infty.
\]

This condition ensures that the previous Risk Margin is finite and strictly positive and permits the normalization introduced below. Finite-horizon formulations considered later are included as a special case.

\subsection{Discounted SCR run-off distribution}
\label{SCRRunoffDistribution}

The revised calibration acts on projected capital according to its position in the run-off. Separating the scale of future capital from its temporal allocation therefore provides a convenient representation of the reform effect.

\begin{definition}[Discounted SCR run-off distribution]
	\label{SCRRunoffDefinition}
	
	Define
	
	\[
	p_t
	=
	\frac{SCR_tD_t}
	{\sum_{s\geq0}SCR_sD_s},
	\qquad t\geq0.
	\]
	
	Let \(\tau\) be the auxiliary integer-valued random variable with probability mass function
	
	\[
	\mathbb{P}(\tau=t)=p_t.
	\]
\end{definition}

The distribution of \(\tau\) is the normalized temporal profile of discounted SCR contributions. It preserves the complete timing structure entering the Risk Margin while removing the overall scale of the projected capital path. Later values of \(\tau\) correspond to a larger relative concentration of discounted capital at distant horizons. This normalization therefore provides a scale-free representation of the timing of the capital required to support the insurance obligations throughout their run-off.

\subsection{Exact representation of the reform effect}
\label{ExactReformEffect}

The normalized run-off distribution yields an exact representation of the mechanical effect of the 2027 calibration.

\begin{proposition}[Exact representation]
	\label{ExactRepresentationProp}
	
	Let
	
	\[
	\kappa
	=
	\frac{c_1}{c_0}
	=
	\frac{0.0475}{0.06}
	=
	\frac{19}{24}.
	\]
	
	Then
	
	\begin{equation}
		\frac{RM^{\mathrm{new}}}
		{RM^{\mathrm{old}}}
		=
		\kappa
		\mathbb{E}\!\left[\lambda_\tau\right].
		\label{ExactRatio}
	\end{equation}
\end{proposition}

\begin{proof}
	By Definition \ref{SCRRunoffDefinition},
	
	\[
	SCR_tD_t
	=
	A p_t.
	\]
	
	Hence
	
	\begin{align*}
		RM^{\mathrm{new}}
		&=
		c_1
		\sum_{t\geq0}
		\lambda_tSCR_tD_t\\
		&=
		c_1A
		\sum_{t\geq0}
		p_t\lambda_t\\
		&=
		c_1A
		\mathbb{E}\!\left[\lambda_\tau\right].
	\end{align*}
	
	Since \(RM^{\mathrm{old}}=c_0A\), division gives
	
	\[
	\frac{RM^{\mathrm{new}}}
	{RM^{\mathrm{old}}}
	=
	\frac{c_1}{c_0}
	\mathbb{E}\!\left[\lambda_\tau\right]
	=
	\kappa
	\mathbb{E}\!\left[\lambda_\tau\right].
	\]
\end{proof}

Equation \eqref{ExactRatio} separates the two channels of the revised calibration. The factor \(\kappa\) captures the reduction in the prescribed cost-of-capital rate, while \(\mathbb{E}[\lambda_\tau]\) captures the interaction between the regulatory time weights and the discounted SCR run-off. The representation therefore shows how the temporal structure of an insurer's projected capital profile determines the relative effect of the revised calibration.

A direct consequence is scale invariance. For any \(a>0\), replacing \(SCR_t\) by \(aSCR_t\) at every date changes both Risk Margins by the same factor and therefore leaves their ratio, and the associated percentage reduction, unchanged. Portfolio differences in the relative reform effect are consequently driven by the temporal allocation of discounted capital rather than by its nominal scale.

Define the total relative reduction as

\[
\Delta
=
1-
\frac{RM^{\mathrm{new}}}
{RM^{\mathrm{old}}},
\]

and write

\[
L
=
\mathbb{E}\!\left[\lambda_\tau\right].
\]

Then

\[
1-\Delta=\kappa L.
\]

The equivalent additive decomposition is

\begin{equation}
	\Delta
	=
	(1-\kappa)
	+
	(1-L)
	-
	(1-\kappa)(1-L).
	\label{AdditiveDecomposition}
\end{equation}

The first two terms correspond to the stand-alone reductions generated by the cost-of-capital rate and the time-dependent factor, respectively, while the final term captures their interaction. The decomposition therefore avoids attributing the multiplicative interaction to either component separately.

The cost-of-capital adjustment alone produces a reduction of

\[
1-\kappa
=
\frac{5}{24}
\simeq20.83\%.
\]

Any further reduction is generated by the time-dependent weighting of the projected capital run-off.

\subsection{Lambda reweighting of the SCR run-off}
\label{LambdaReweighting}

The revised calibration also changes the relative contribution of different dates within the Risk Margin. Define

\begin{equation}
	p_t^{\lambda}
	=
	\frac{\lambda_t p_t}
	{\mathbb{E}[\lambda_\tau]},
	\qquad t\geq0,
	\label{LambdaWeightedProbabilities}
\end{equation}

and let \(\tau^\lambda\) denote a random variable with this probability mass function. Equivalently,

\[
p_t^\lambda
=
\frac{c_1\lambda_tSCR_tD_t}
{RM^{\mathrm{new}}},
\]

so that \(p_t^\lambda\) is the normalized share of the revised Risk Margin attributable to projection date \(t\). Since \(\lambda_t>0\), \(\tau^\lambda\) and \(\tau\) have the same support. On that support,

\begin{equation}
	\frac{p_t^\lambda}{p_t}
	=
	\frac{\lambda_t}
	{\mathbb{E}[\lambda_\tau]}.
	\label{LikelihoodRatioTilt}
\end{equation}

Because \(t\mapsto\lambda_t\) is non-increasing, the likelihood ratio in \eqref{LikelihoodRatioTilt} is non-increasing in \(t\). For discrete random variables \(X\) and \(Y\) with common support, we use the convention \(X\leq_{\mathrm{lr}}Y\) when the ratio of their probability masses \(\mathbb{P}(X=t)/\mathbb{P}(Y=t)\) is non-increasing in \(t\).

\begin{proposition}[Temporal effect of lambda reweighting]
	\label{LambdaReweightingProp}
	
	The lambda-weighted run-off distribution is earlier than the original discounted SCR run-off distribution in likelihood-ratio order and therefore in usual stochastic order,
	
	\[
	\tau^\lambda
	\leq_{\mathrm{lr}}
	\tau
	\quad\Longrightarrow\quad
	\tau^\lambda
	\leq_{\mathrm{st}}
	\tau.
	\]
\end{proposition}

\begin{proof}
	Equation \eqref{LikelihoodRatioTilt} shows that the ratio of the probability masses of \(\tau^\lambda\) and \(\tau\) is non-increasing over their common support. Hence
	\(\tau^\lambda\leq_{\mathrm{lr}}\tau\).
	The implication from likelihood-ratio order to usual stochastic order follows from the standard hierarchy of stochastic orders; see Shaked and Shanthikumar (2007)\cite{shaked2007stochastic}.
\end{proof}

The proposition formalizes the temporal reallocation induced by the revised formula within a fixed projected SCR path. Relative to the previous calculation, lambda weighting shifts the normalized distribution of Risk Margin contributions towards earlier dates. Whenever the support contains positive mass at dates associated with different values of \(\lambda_t\), the reweighted distribution differs from the original one and the ordering is non-trivial.

The \(50\%\) floor limits the extent of this reallocation. Since \(\lambda_t=0.50\) for all \(t\geq17\), the lambda mechanism no longer distinguishes among dates within the floor region. Conditional on the total discounted SCR mass assigned to that region, redistributing this mass among dates \(t\geq17\) leaves \(\mathbb{E}[\lambda_\tau]\), and therefore the reform ratio in \eqref{ExactRatio}, unchanged. For long-duration insurance portfolios, the lambda channel therefore depends on the aggregate discounted SCR mass beyond year 17, whereas interest-rate sensitivity continues to depend on its allocation across dates.

\section{SCR run-off structure and bounds}
\label{SCRRunoffBounds}

The representation introduced in Section \ref{AnalyticalFramework} makes the temporal structure of discounted capital requirements central to the analysis of the reform. This section derives bounds on the reform effect, characterizes its dependence on the location and dispersion of the SCR run-off, and examines how proportional Best Estimate projections interact with the new time-dependent weighting.

\subsection{Horizon-specific and universal bounds}
\label{HorizonBounds}

Suppose first that the discounted SCR run-off distribution has finite support contained in

\[
\{0,1,\ldots,T\}.
\]

Since \(t\mapsto\lambda_t\) is non-increasing,

\[
\lambda_T
\leq
\mathbb{E}[\lambda_\tau]
\leq
\lambda_0=1.
\]

Combining this inequality with Proposition \ref{ExactRepresentationProp} gives the following result.

\begin{proposition}[Horizon-specific bounds]
	\label{HorizonBoundsProp}
	
	If
	
	\[
	\mathbb{P}(0\leq\tau\leq T)=1,
	\]
	
	then
	
	\begin{equation}
		\kappa\lambda_T
		\leq
		\frac{RM^{\mathrm{new}}}
		{RM^{\mathrm{old}}}
		\leq
		\kappa.
		\label{HorizonRatioBounds}
	\end{equation}
	
	Equivalently,
	
	\[
	1-\kappa
	\leq
	\Delta
	\leq
	1-\kappa\lambda_T.
	\]
	
	Both bounds are sharp.
\end{proposition}

\begin{proof}
	Since \(0\leq\tau\leq T\) almost surely,
	
	\[
	\lambda_T
	\leq
	\lambda_\tau
	\leq
	1.
	\]
	
	Taking expectations and applying Proposition \ref{ExactRepresentationProp} gives the stated inequalities. The upper bound for the ratio is attained when all discounted SCR mass is concentrated at \(t=0\), while the lower bound is attained when all mass is concentrated at \(t=T\).
\end{proof}

Since

\[
0.50
\leq
\lambda_t
\leq
1
\]

for every \(t\), the horizon-specific result yields portfolio-independent bounds.

\begin{corollary}[Universal bounds]
	\label{UniversalBoundsCor}
	
	For every admissible projected SCR profile,
	
	\begin{equation}
		\begin{aligned}
			\frac{19}{48}
			&\leq
			\frac{RM^{\mathrm{new}}}
			{RM^{\mathrm{old}}}
			\leq
			\frac{19}{24},
			\\[1mm]
			\frac{5}{24}
			&\leq
			\Delta
			\leq
			\frac{29}{48}.
		\end{aligned}
		\label{UniversalBounds}
	\end{equation}
	
	Thus, the mechanical reduction lies between approximately \(20.83\%\) and \(60.42\%\).
\end{corollary}

The upper bound on \(\Delta\) combines the \(50\%\) floor on the lambda factor with the simultaneous reduction in the prescribed cost-of-capital rate. Accordingly, the floor constrains the attenuation applied to future SCRs, while the total percentage reduction reflects both components of the revised calibration.

If the capital run-off terminates before the floor becomes operative, the horizon-specific bound is tighter. For \(T<17\),

\[
\lambda_T=0.96^T,
\]

and hence

\[
\Delta
\leq
1-
\frac{19}{24}0.96^T.
\]

\subsection{Run-off timing under stochastic dominance}
\label{RunoffStochasticDominance}

The horizon bounds describe extreme temporal concentrations of discounted capital. Comparisons between general run-off profiles follow directly from stochastic ordering.

Consider two portfolios \(A\) and \(B\), and construct their discounted SCR timing variables \(\tau_A\) and \(\tau_B\) from their respective normalized profiles as in Definition \ref{SCRRunoffDefinition}.

\begin{proposition}[Later SCR run-off]
	\label{LaterSCRRunoffProp}
	
	If
	
	\[
	\tau_A
	\leq_{\mathrm{st}}
	\tau_B,
	\]
	
	then
	
	\[
	\mathbb{E}[\lambda_{\tau_B}]
	\leq
	\mathbb{E}[\lambda_{\tau_A}],
	\]
	
	and therefore
	
	\[
	\frac{RM_B^{\mathrm{new}}}
	{RM_B^{\mathrm{old}}}
	\leq
	\frac{RM_A^{\mathrm{new}}}
	{RM_A^{\mathrm{old}}},
	\qquad
	\Delta_B
	\geq
	\Delta_A.
	\]
\end{proposition}

\begin{proof}
	Usual stochastic order implies
	
	\[
	\mathbb{E}[f(\tau_A)]
	\geq
	\mathbb{E}[f(\tau_B)]
	\]
	
	for every bounded non-increasing function \(f\). Since \(t\mapsto\lambda_t\) is bounded and non-increasing, choosing \(f(t)=\lambda_t\) gives the result. See Shaked and Shanthikumar (2007)\cite{shaked2007stochastic}.
\end{proof}

A later discounted SCR run-off therefore produces a weakly larger percentage reduction under the revised calibration. The ordering is defined by the timing of discounted capital and may differ from an ordering based on contractual maturity. Portfolios with similar liability maturities may consequently experience different reform effects when their non-hedgeable risks release capital at different speeds.

Equality occurs, in particular, when timing changes are confined to dates over which the lambda factor is constant.

\subsection{EIOPA Best Estimate approximation}
\label{EIOPAMethod2}

Projecting future SCRs is one of the operationally demanding components of the Risk Margin calculation. Dreksler et al. (2015)\cite{dreksler2015solvency} discuss a hierarchy of simplified approaches, including proportional projections based on the run-off of Best Estimate technical provisions. Ohlsson and Lauzeningks (2009)\cite{ohlsson2009oneyear} similarly relate simplified cost-of-capital calculations to the run-off of insurance reserves.

The revised EIOPA Guidelines retain proportional Best Estimate projection as Method 2 \cite{eiopa2026valuation}. In the notation used here,

\begin{equation}
	SCR_t^{\mathrm{M2}}
	=
	SCR_0
	\frac{BE_t}{BE_0},
	\label{Method2SCR}
\end{equation}

where \(BE_t\) denotes the projected Best Estimate technical provisions net of reinsurance.

The approximation relies on sufficient stability of the risk profile over time, including the composition of risks and sub-risks and the relationship between their main drivers and the net Best Estimate. EIOPA also specifies that Method 2 is inappropriate when negative Best Estimate values arise at the valuation date or at subsequent dates \cite{eiopa2026valuation}.

Best Estimate and SCR nevertheless represent different objects. The former represents the probability-weighted present value of future liability cash flows, while the latter measures adverse one-year changes in the solvency position. Their run-off patterns can consequently diverge. Dreksler et al. (2015)\cite{dreksler2015solvency} distinguish risk components that may evolve broadly in line with Best Estimate provisions from components, such as premium or catastrophe risk, that may follow different dynamics. England et al. (2019)\cite{england2019lifetime} likewise construct alternative capital profiles based on the Best Estimate, standard deviation and \(99.5\%\) Value-at-Risk of future claims development results.

\subsection{Capital intensity and Best Estimate proportionality}
\label{CapitalIntensityMethod2}

Assume throughout this subsection that

\[
BE_t>0
\]

at every relevant date and that the required discounted sums are finite. Define the capital intensity

\begin{equation}
	h_t
	=
	\frac{SCR_t}{BE_t}.
	\label{CapitalIntensity}
\end{equation}

Exact Method 2 proportionality corresponds to

\[
h_t=h_0
\]

for every \(t\).

Let

\[
B
=
\sum_{t\geq0}
BE_tD_t
\]

and define the discounted Best Estimate timing distribution by

\[
\pi_t
=
\frac{BE_tD_t}{B}.
\]

Under \(\mathbb{E}_\pi\), the time index \(\tau\) has probability mass function \((\pi_t)\). Expectations and covariances under this distribution are denoted by \(\mathbb{E}_\pi\) and \(\operatorname{Cov}_\pi\), respectively. The full Risk Margins satisfy

\[
RM_{\mathrm{full}}^{\mathrm{old}}
=
c_0B
\mathbb{E}_{\pi}[h_\tau]
\]

and

\[
RM_{\mathrm{full}}^{\mathrm{new}}
=
c_1B
\mathbb{E}_{\pi}[\lambda_\tau h_\tau].
\]

Under Method 2,

\[
RM_{\mathrm{M2}}^{\mathrm{old}}
=
c_0Bh_0,
\qquad
RM_{\mathrm{M2}}^{\mathrm{new}}
=
c_1Bh_0
\mathbb{E}_{\pi}[\lambda_\tau].
\]

The corresponding new-to-old Risk Margin ratios are

\begin{equation}
	\begin{aligned}
		R_{\mathrm{full}}
		&:=
		\frac{RM_{\mathrm{full}}^{\mathrm{new}}}
		{RM_{\mathrm{full}}^{\mathrm{old}}}
		=
		\kappa
		\frac{
			\mathbb{E}_{\pi}[\lambda_\tau h_\tau]
		}{
			\mathbb{E}_{\pi}[h_\tau]
		},
		\\[1mm]
		R_{\mathrm{M2}}
		&:=
		\frac{RM_{\mathrm{M2}}^{\mathrm{new}}}
		{RM_{\mathrm{M2}}^{\mathrm{old}}}
		=
		\kappa
		\mathbb{E}_{\pi}[\lambda_\tau].
	\end{aligned}
	\label{FullMethod2Ratios}
\end{equation}

Define the corresponding relative reductions by

\[
\Delta_{\mathrm{full}}
=
1-R_{\mathrm{full}},
\qquad
\Delta_{\mathrm{M2}}
=
1-R_{\mathrm{M2}}.
\]

The difference between these ratios has an exact covariance representation.

\begin{proposition}[Capital-intensity interaction]
	\label{CapitalIntensityProp}
	
	Under the preceding assumptions,
	
	\begin{equation}
		R_{\mathrm{full}}
		-
		R_{\mathrm{M2}}
		=
		\kappa
		\frac{
			\operatorname{Cov}_{\pi}
			(\lambda_\tau,h_\tau)
		}{
			\mathbb{E}_{\pi}[h_\tau]
		}.
		\label{CapitalIntensityCovariance}
	\end{equation}
\end{proposition}

\begin{proof}
	Using
	
	\[
	\mathbb{E}_{\pi}[\lambda_\tau h_\tau]
	=
	\mathbb{E}_{\pi}[\lambda_\tau]
	\mathbb{E}_{\pi}[h_\tau]
	+
	\operatorname{Cov}_{\pi}(\lambda_\tau,h_\tau),
	\]
	
	the first identity in \eqref{FullMethod2Ratios} gives
	
	\[
	R_{\mathrm{full}}
	=
	R_{\mathrm{M2}}
	+
	\kappa
	\frac{
		\operatorname{Cov}_{\pi}(\lambda_\tau,h_\tau)
	}{
		\mathbb{E}_{\pi}[h_\tau]
	},
	\]
	
	which proves the result.
\end{proof}

The covariance in \eqref{CapitalIntensityCovariance} isolates the interaction between departures from Best Estimate proportionality and the revised temporal weighting. The relative effect of the reform is governed by the association between capital intensity \(SCR_t/BE_t\) and the lambda sequence.

\begin{corollary}[Monotone capital intensity]
	\label{MonotoneCapitalIntensityCor}
	
	If \(h_t\) is non-decreasing in \(t\), then
	
	\[
	R_{\mathrm{full}}
	\leq
	R_{\mathrm{M2}},
	\qquad
	\Delta_{\mathrm{full}}
	\geq
	\Delta_{\mathrm{M2}}.
	\]
	
	If \(h_t\) is non-increasing, the inequalities are reversed.
\end{corollary}

\begin{proof}
	Since \((\lambda_t)\) is non-increasing, a non-decreasing sequence \((h_t)\) satisfies
	
	\[
	\operatorname{Cov}_{\pi}
	(\lambda_\tau,h_\tau)
	\leq0.
	\]
	
	Proposition \ref{CapitalIntensityProp} then gives the result. The argument is reversed when \(h_t\) is non-increasing.
\end{proof}

When capital intensity rises over time, the full SCR profile places relatively more weight on later dates than its Method 2 counterpart. Method 2 then understates the percentage reduction generated by the reform. A declining capital intensity produces the opposite ordering.

These inequalities characterize the relative reform effect. The approximation error in the absolute Risk Margin additionally depends on the level of capital intensity. Specifically,

\begin{equation}
	RM_{\mathrm{full}}^{\mathrm{new}}
	-
	RM_{\mathrm{M2}}^{\mathrm{new}}
	=
	c_1B
	\left[
	\mathbb{E}_{\pi}[\lambda_\tau]
	\left(
	\mathbb{E}_{\pi}[h_\tau]-h_0
	\right)
	+
	\operatorname{Cov}_{\pi}(\lambda_\tau,h_\tau)
	\right].
	\label{AbsoluteMethod2Error}
\end{equation}
The absolute error therefore combines a level component and a timing component, with its sign determined jointly by the two terms.

\subsection{Mean run-off time and dispersion}
\label{MeanRunoffDispersion}

Assume that the discounted SCR timing variable has a finite first moment and define

\[
\mu
=
\mathbb{E}[\tau],
\]
referred to as the \emph{mean discounted-SCR run-off time}. This quantity is distinct from EIOPA's duration-based Risk Margin approximation (Method 3), which relies on an adjusted modified duration of insurance liabilities and serves a different approximation purpose \cite{eiopa2026valuation}.

The mean summarizes the location of the discounted capital run-off, while the reform effect also depends on its temporal dispersion. Extend the regulatory factor to \(s\geq0\) by

\[
\lambda(s)
=
\max\left(0.96^s,0.50\right).
\]

Since \(s\mapsto0.96^s\) and the constant function \(0.50\) are convex, their pointwise maximum is convex. Hence \(\lambda\) is convex and non-increasing.

\begin{proposition}[Dispersion at fixed mean]
	\label{ConvexOrderProp}
	
	Let \(\tau_A\) and \(\tau_B\) have the same finite mean. If
	
	\[
	\tau_A \leq_{\mathrm{cx}} \tau_B
	\]
	
	in convex order, then
	
	\[
	\mathbb{E}[\lambda(\tau_A)]
	\leq
	\mathbb{E}[\lambda(\tau_B)].
	\]
	
	Consequently,
	
	\[
	\frac{RM_A^{\mathrm{new}}}
	{RM_A^{\mathrm{old}}}
	\leq
	\frac{RM_B^{\mathrm{new}}}
	{RM_B^{\mathrm{old}}},
	\qquad
	\Delta_A
	\geq
	\Delta_B.
	\]
\end{proposition}

\begin{proof}
	Convex order implies
	
	\[
	\mathbb{E}[\varphi(\tau_A)]
	\leq
	\mathbb{E}[\varphi(\tau_B)]
	\]
	
	for every convex function \(\varphi\) for which the expectations exist. Taking \(\varphi=\lambda\) gives the first inequality, since \(\lambda(\tau)=\lambda_\tau\) for integer-valued \(\tau\). The remaining inequalities follow from Proposition \ref{ExactRepresentationProp}. See Shaked and Shanthikumar (2007)\cite{shaked2007stochastic}.
\end{proof}

At a fixed mean discounted-SCR run-off time, greater dispersion in the convex-order sense produces a smaller percentage reduction under the reform. Both mean timing and temporal dispersion therefore shape the effect of the lambda factor.

Equality arises, for example, when discounted SCR mass is redistributed exclusively among dates \(t\geq17\), where \(\lambda_t=0.50\).

\subsection{Sharp bounds conditional on mean run-off time}
\label{SharpMeanBounds}

Suppose

\[
\tau\in\{0,1,\ldots,T\},
\qquad
T\geq1,
\]

with

\[
\mathbb{E}[\tau]=\mu,
\qquad
0\leq\mu\leq T.
\]

If \(\mu=T\), then necessarily \(\tau=T\) almost surely and

\[
\mathbb{E}[\lambda_\tau]=\lambda_T.
\]

Consider therefore \(0\leq\mu<T\), and set

\[
m=\lfloor\mu\rfloor,
\qquad
\theta=\mu-m.
\]

\begin{theorem}[Sharp bounds with fixed mean and horizon]
	\label{SharpMeanBoundsThm}
	
	Under the preceding assumptions,
	
	\begin{equation}
		(1-\theta)\lambda_m
		+
		\theta\lambda_{m+1}
		\leq
		\mathbb{E}[\lambda_\tau]
		\leq
		\left(1-\frac{\mu}{T}\right)\lambda_0
		+
		\frac{\mu}{T}\lambda_T.
		\label{SharpLambdaBounds}
	\end{equation}
	
	The lower bound is attained by
	
	\[
	\mathbb{P}(\tau=m)=1-\theta,
	\qquad
	\mathbb{P}(\tau=m+1)=\theta,
	\]
	
	while the upper bound is attained by
	
	\[
	\mathbb{P}(\tau=0)
	=
	1-\frac{\mu}{T},
	\qquad
	\mathbb{P}(\tau=T)
	=
	\frac{\mu}{T}.
	\]
\end{theorem}

\begin{proof}
	Let \(\widetilde{\lambda}\) denote the piecewise-linear interpolation through
	
	\[
	(t,\lambda_t),
	\qquad
	t=0,\ldots,T.
	\]
	
	The discrete sequence \((\lambda_t)\) is convex, so \(\widetilde{\lambda}\) is convex on \([0,T]\), with
	
	\[
	\widetilde{\lambda}(t)=\lambda_t
	\]
	
	at integer dates. Jensen's inequality therefore gives
	
	\[
	\mathbb{E}[\lambda_\tau]
	=
	\mathbb{E}[\widetilde{\lambda}(\tau)]
	\geq
	\widetilde{\lambda}(\mathbb{E}[\tau]).
	\]
	
	Since \(\mu=m+\theta\),
	
	\[
	\widetilde{\lambda}(\mu)
	=
	(1-\theta)\lambda_m
	+
	\theta\lambda_{m+1},
	\]
	
	and the adjacent-point distribution stated in the theorem attains this value.
	
	For the upper bound, convexity places \(\widetilde{\lambda}\) below the chord joining its endpoint values,
	
	\[
	\widetilde{\lambda}(t)
	\leq
	\left(1-\frac{t}{T}\right)\lambda_0
	+
	\frac{t}{T}\lambda_T.
	\]
	
	Taking expectations yields
	
	\[
	\mathbb{E}[\lambda_\tau]
	\leq
	\left(1-\frac{\mu}{T}\right)\lambda_0
	+
	\frac{\mu}{T}\lambda_T.
	\]
	
	The endpoint distribution stated in the theorem attains the upper bound.
\end{proof}

The extremal distributions provide a convex-order interpretation of the result. Among integer-valued distributions on \(\{0,\ldots,T\}\) with mean \(\mu\), concentration on the adjacent integers surrounding \(\mu\) minimizes dispersion, while concentration on the endpoints maximizes it. Related extremal bounds for expectations of convex functions arise in the Jensen and Edmundson-Madansky framework; see Huang et al. (1977)\cite{huang1977bounds}.

Combining \eqref{SharpLambdaBounds} with Proposition \ref{ExactRepresentationProp} gives
\begin{equation}
	\kappa
	\left[
	(1-\theta)\lambda_m+\theta\lambda_{m+1}
	\right]
	\leq
	\frac{RM^{\mathrm{new}}}{RM^{\mathrm{old}}}
	\leq
	\kappa
	\left[
	\left(1-\frac{\mu}{T}\right)\lambda_0
	+
	\frac{\mu}{T}\lambda_T
	\right].
	\label{SharpRatioBounds}
\end{equation}
The corresponding bounds for \(\Delta\) follow by subtracting from one and reversing the inequalities.

\subsection{Geometric run-off benchmark}
\label{GeometricRunoff}

A geometric capital profile provides a one-parameter benchmark for isolating run-off persistence. Let

\begin{equation}
	SCR_t(q)
	=
	SCR_0q^t,
	\qquad
	t=0,\ldots,T,
	\qquad
	0\leq q<1.
	\label{GeometricSCR}
\end{equation}

Higher values of \(q\) correspond to a slower release of regulatory capital. For a finite horizon \(T\) and fixed positive discount factors,

\[
p_t(q)
=
\frac{q^tD_t}
{\sum_{s=0}^{T}q^sD_s}.
\]

Expectations and covariances under this distribution are denoted by \(\mathbb{E}_q\) and \(\operatorname{Cov}_q\), respectively.

\begin{proposition}[Persistence monotonicity]
	\label{PersistenceMonotonicityProp}
	
	The relative reduction
	
	\[
	\Delta(q)
	=
	1-
	\frac{RM^{\mathrm{new}}(q)}
	{RM^{\mathrm{old}}(q)}
	\]
	
	is non-decreasing in \(q\).
\end{proposition}

\begin{proof}
	For \(q>0\),
	
	\[
	\frac{\partial}
	{\partial\log q}
	\mathbb{E}_q[\lambda_\tau]
	=
	\operatorname{Cov}_q(\lambda_\tau,\tau).
	\]
	
	Since \(t\mapsto\lambda_t\) is non-increasing and \(t\mapsto t\) is increasing,
	
	\[
	\operatorname{Cov}_q(\lambda_\tau,\tau)
	\leq0.
	\]
	
	Hence
	
	\[
	\frac{\partial}
	{\partial q}
	\mathbb{E}_q[\lambda_\tau]
	\leq0,
	\]
	
	and Proposition \ref{ExactRepresentationProp} gives
	
	\[
	\Delta'(q)\geq0.
	\]
	
	The result at \(q=0\) follows by continuity.
\end{proof}

Under a flat annual effective risk-free rate \(r\), define

\[
z
=
\frac{q}{1+r},
\qquad
G_n(z)
=
\sum_{j=0}^{n}z^j.
\]

The corresponding Risk Margins are

\begin{equation}
	\begin{aligned}
		RM^{\mathrm{old}}
		&=
		\frac{c_0SCR_0}{1+r}
		G_T(z),
		\\[1mm]
		RM^{\mathrm{new}}
		&=
		\begin{cases}
			\dfrac{c_1SCR_0}{1+r}
			G_T(0.96z),
			& T\leq16,
			\\[3mm]
			\dfrac{c_1SCR_0}{1+r}
			\left[
			G_{16}(0.96z)
			+
			0.50\,z^{17}G_{T-17}(z)
			\right],
			& T\geq17.
		\end{cases}
	\end{aligned}
	\label{GeometricClosedForms}
\end{equation}

These expressions provide the controlled benchmark used later to quantify the effect of capital persistence while holding scale and other run-off features fixed.
\section{Interest-rate sensitivity}
\label{InterestRateSensitivity}

A longstanding criticism of the Solvency II Risk Margin concerns its sensitivity to interest rates, particularly for long-duration insurance liabilities. Pelkiewicz et al. (2020)\cite{pelkiewicz2020review} discuss this issue extensively, while Waszink (2024)\cite{waszink2024comments} examines the direct sensitivity of alternative Risk Margin specifications under common projected SCR profiles. The revised EIOPA Guidelines identify lower interest-rate sensitivity as one of the intended effects of the new time-dependent factor \cite{eiopa2026valuation}.

The analysis separates the direct discounting channel, evaluated under a fixed projected \(SCR_t\) sequence, from the additional response generated when interest-rate movements alter future SCRs. The two channels are examined in turn.

\subsection{Direct discount-rate sensitivity}
\label{DirectRateSensitivity}

Let \(x\) denote a parallel continuously compounded shift applied to a baseline discount structure. Define

\begin{equation}
	D_t(x)
	=
	D_t(0)
	\exp\!\left(-(t+1)x\right),
	\qquad t\geq0.
	\label{ShiftedDiscount}
\end{equation}

Set

\[
u_t=t+1.
\]

The previous and revised Risk Margins become

\[
RM^{\mathrm{old}}(x)
=
c_0
\sum_{t\geq0}
SCR_tD_t(0)e^{-u_tx}
\]

and

\[
RM^{\mathrm{new}}(x)
=
c_1
\sum_{t\geq0}
\lambda_tSCR_tD_t(0)e^{-u_tx}.
\]

Assume that the relevant sums and first-moment sums are finite and that termwise differentiation is justified in a neighbourhood of the value of \(x\) under consideration. Define the \(x\)-dependent discounted SCR distribution by

\[
p_t(x)
=
\frac{
	SCR_tD_t(0)e^{-u_tx}
}{
	\sum_{s\geq0}
	SCR_sD_s(0)e^{-u_sx}
},
\qquad t\geq0.
\]

Under \(\mathbb{E}_x\), the time index \(\tau\) has probability mass function \((p_t(x))\). Expectations and covariances under this distribution are denoted by \(\mathbb{E}_x\) and \(\operatorname{Cov}_x\), respectively.

Define the direct interest-rate semi-elasticities as

\begin{equation}
	\mathcal{D}_{\mathrm{old}}(x)
	=
	-
	\frac{\partial}{\partial x}
	\log RM^{\mathrm{old}}(x),
	\qquad
	\mathcal{D}_{\mathrm{new}}(x)
	=
	-
	\frac{\partial}{\partial x}
	\log RM^{\mathrm{new}}(x).
	\label{DirectSemiElasticities}
\end{equation}

Direct differentiation gives

\[
\mathcal{D}_{\mathrm{old}}(x)
=
\mathbb{E}_x[u_\tau],
\qquad
\mathcal{D}_{\mathrm{new}}(x)
=
\frac{
	\mathbb{E}_x[u_\tau\lambda_\tau]
}{
	\mathbb{E}_x[\lambda_\tau]
}.
\]

\subsection{Lambda reweighting and sensitivity reduction}
\label{LambdaSensitivityReduction}

The second expression is the expectation of \(u_\tau\) after applying the same lambda reweighting introduced in Section \ref{LambdaReweighting}, now to the \(x\)-dependent discounted SCR distribution.

\begin{proposition}[Reduction in direct interest-rate sensitivity]
	\label{RateSensitivityReductionProp}
	
	For every \(x\) satisfying the preceding regularity conditions,
	
	\begin{equation}
		\mathcal{D}_{\mathrm{new}}(x)
		-
		\mathcal{D}_{\mathrm{old}}(x)
		=
		\frac{
			\operatorname{Cov}_x
			(u_\tau,\lambda_\tau)
		}{
			\mathbb{E}_x[\lambda_\tau]
		}
		\leq0.
		\label{SensitivityCovariance}
	\end{equation}
	
	Hence
	
	\[
	\mathcal{D}_{\mathrm{new}}(x)
	\leq
	\mathcal{D}_{\mathrm{old}}(x).
	\]
\end{proposition}

\begin{proof}
	Using the expectation representations above,
	
	\begin{align*}
		\mathcal{D}_{\mathrm{new}}(x)
		-
		\mathcal{D}_{\mathrm{old}}(x)
		&=
		\frac{
			\mathbb{E}_x[u_\tau\lambda_\tau]
		}{
			\mathbb{E}_x[\lambda_\tau]
		}
		-
		\mathbb{E}_x[u_\tau]
		\\
		&=
		\frac{
			\operatorname{Cov}_x
			(u_\tau,\lambda_\tau)
		}{
			\mathbb{E}_x[\lambda_\tau]
		}.
	\end{align*}
	
	Since \(u_t=t+1\) is increasing in \(t\) and \(\lambda_t\) is non-increasing,
	
	\[
	\operatorname{Cov}_x
	(u_\tau,\lambda_\tau)
	\leq0.
	\]
	
	Moreover, \(\mathbb{E}_x[\lambda_\tau]>0\), which proves the result.
\end{proof}

The covariance identity shows that lambda reweighting concentrates the normalized Risk Margin contributions at earlier dates and reduces the direct semi-elasticity with respect to the imposed parallel rate shift. From an actuarial perspective, the revised formula thereby reduces the effective timing of the regulatory capital costs embedded in technical provisions. The inequality is strict whenever the support of the discounted SCR distribution contains positive probability at dates associated with distinct lambda factors.

Equality arises whenever \(\lambda_t\) is constant on the relevant support. In particular, if all discounted SCR contributions occur at dates \(t\geq17\), then \(\lambda_t=0.50\) throughout the support and \(\mathcal{D}_{\mathrm{new}}(x)=\mathcal{D}_{\mathrm{old}}(x)\). Within the floor region, both calibrations therefore have the same direct semi-elasticity.

\subsection{Sensitivity of the reform effect}
\label{ReformRateSensitivity}

The same covariance structure determines how the relative effect of the reform responds to the imposed rate shift. Define

\[
R(x)
=
\frac{RM^{\mathrm{new}}(x)}
{RM^{\mathrm{old}}(x)},
\qquad
\Delta(x)
=
1-R(x).
\]

From Proposition \ref{ExactRepresentationProp},

\[
R(x)
=
\kappa
\mathbb{E}_x[\lambda_\tau].
\]

\begin{proposition}[Interest-rate sensitivity of the reform ratio]
	\label{ReformRatioRateProp}
	
	Under the preceding assumptions,
	
	\begin{equation}
		R'(x)
		=
		-
		\kappa
		\operatorname{Cov}_x
		(\lambda_\tau,u_\tau)
		\geq0,
		\qquad
		\Delta'(x)\leq0.
		\label{ReformRateDerivative}
	\end{equation}
\end{proposition}

\begin{proof}
	For any function \(f\) satisfying the required integrability conditions,
	
	\[
	\frac{\partial}{\partial x}
	\mathbb{E}_x[f(\tau)]
	=
	-
	\operatorname{Cov}_x
	(f(\tau),u_\tau).
	\]
	
	Taking \(f(\tau)=\lambda_\tau\) gives
	
	\[
	R'(x)
	=
	-
	\kappa
	\operatorname{Cov}_x
	(\lambda_\tau,u_\tau).
	\]
	
	Since \(\lambda_t\) is non-increasing and \(u_t\) is increasing,
	
	\[
	\operatorname{Cov}_x
	(\lambda_\tau,u_\tau)
	\leq0.
	\]
	
	Hence \(R'(x)\geq0\). Since
	
	\[
	\Delta(x)=1-R(x),
	\]
	
	it follows that \(\Delta'(x)\leq0\).
\end{proof}

A positive parallel rate shift therefore increases the ratio of the revised to the previous Risk Margin and reduces the percentage reduction generated by the reform. Higher discount rates reduce the present value of distant capital contributions more strongly, while those same contributions receive the smallest lambda factors. Their diminished weight consequently makes the two calibrations relatively closer.

Conversely, lower discount rates preserve more of the present value of distant SCR contributions and strengthen the relative effect of the time-dependent attenuation. This mechanism is consistent with the numerical behaviour documented for long-duration portfolios by Waszink (2024)\cite{waszink2024comments}.

\begin{remark}[Rate-dependent SCR profiles]
	\label{RateDependentSCRRemark}
	
	The total response can be characterized when the projected SCR path depends on the imposed rate shift. Let
	
	\[
	SCR_t=SCR_t(x)>0
	\]
	
	on the relevant support, and define
	
	\[
	\eta_t(x)
	=
	\frac{\partial}{\partial x}
	\log SCR_t(x),
	\qquad
	m_t(x)
	=
	(t+1)-\eta_t(x).
	\]
	
	The relevant discounted SCR distribution is then
	
	\[
	p_t(x)
	=
	\frac{
		SCR_t(x)D_t(0)e^{-u_tx}
	}{
		\sum_{s\geq0}
		SCR_s(x)D_s(0)e^{-u_sx}
	}.
	\]
	
	Under suitable differentiability and summability conditions, define the total semi-elasticities by
	
	\[
	\mathcal{D}_{\mathrm{old}}^{\mathrm{tot}}(x)
	=
	-
	\frac{\partial}{\partial x}
	\log RM^{\mathrm{old}}(x),
	\qquad
	\mathcal{D}_{\mathrm{new}}^{\mathrm{tot}}(x)
	=
	-
	\frac{\partial}{\partial x}
	\log RM^{\mathrm{new}}(x).
	\]
	
	Since each discounted capital term has logarithmic semi-elasticity \(m_t(x)\), the same differentiation argument gives
	
	\begin{equation}
		\mathcal{D}_{\mathrm{new}}^{\mathrm{tot}}(x)
		-
		\mathcal{D}_{\mathrm{old}}^{\mathrm{tot}}(x)
		=
		\frac{
			\operatorname{Cov}_x
			(m_\tau(x),\lambda_\tau)
		}{
			\mathbb{E}_x[\lambda_\tau]
		}.
		\label{TotalSemiElasticityDifference}
	\end{equation}
	
	A sufficient condition for the revised Risk Margin to retain a lower total semi-elasticity is that \(m_t(x)\) be non-decreasing in \(t\). The ordering of total semi-elasticities is therefore determined jointly by the temporal capital response and the regulatory lambda sequence. Waszink (2024)\cite{waszink2024comments} similarly distinguishes direct discount-rate sensitivity under an unchanged SCR profile from the broader response in which projected capital requirements themselves vary with interest rates.
\end{remark}

\section{Stochastic run-off persistence}
\label{StochasticRunoffPersistence}

The preceding analysis treats the projected SCR path as fixed at the valuation date. Future capital requirements are, however, state-dependent. M\"ohr (2011)\cite{mohr2011market} observes that the SCR applicable at a future date depends on the state prevailing at that date and therefore on information unavailable at time zero, while Christiansen and Niemeyer (2014)\cite{christiansen2014fundamental} analyse alternative mathematical definitions of the SCR within a multi-period framework. Pelkiewicz et al. (2020)\cite{pelkiewicz2020review} likewise emphasize the practical difficulty of projecting future capital requirements over the full liability run-off.

This section adopts a reduced-form representation of uncertainty in one feature of the projected SCR path before application of the regulatory time-dependent factor, namely its persistence. The model separates expected persistence from uncertainty around that persistence and quantifies how the revised calibration transmits both components to the Risk Margin. Its purpose is diagnostic and focuses on the temporal persistence channel.

\subsection{Reduced-form stochastic persistence model}
\label{StochasticPersistenceModel}

Let \(SCR_0>0\) and let \(Q\in[0,1]\) be a random path-level persistence parameter representing uncertainty at the valuation date about the rate at which the projected SCR run-off decays. Higher values of \(Q\) correspond to a slower release of capital over the run-off. Define

\begin{equation}
	SCR_t(Q)
	=
	SCR_0Q^t,
	\qquad
	t=0,\ldots,T.
	\label{StochasticSCRPersistence}
\end{equation}

The sequence \((SCR_t(Q))\) describes the projected capital path before regulatory time weighting, while \((\lambda_t)\) remains deterministic. Conditional on \(Q=q\),

\begin{equation}
	\begin{aligned}
		RM^{\mathrm{old}}(q)
		&=
		c_0SCR_0
		\sum_{t=0}^{T}
		D_tq^t,
		\\
		RM^{\mathrm{new}}(q)
		&=
		c_1SCR_0
		\sum_{t=0}^{T}
		\lambda_tD_tq^t.
	\end{aligned}
	\label{RiskMarginPersistence}
\end{equation}

Throughout this section, expectations with respect to \(Q\) are scenario averages used to quantify uncertainty in the projected capital run-off. For each realization \(Q=q\), the Risk Margin is the conditional regulatory calculation generated by the corresponding projected SCR path.

\subsection{Convexity and Jensen effects}
\label{JensenEffects}

The dependence of each Risk Margin on the persistence parameter follows directly from its first two derivatives.

\begin{proposition}[Monotonicity and convexity in persistence]
	\label{PersistenceConvexityProp}
	
	Both mappings
	\(q\mapsto RM^{\mathrm{old}}(q)\) and
	\(q\mapsto RM^{\mathrm{new}}(q)\)
	are non-decreasing and convex on \([0,1]\). If \(T\geq2\), both mappings are strictly convex.
\end{proposition}

\begin{proof}
	For the previous calibration,
	
	\[
	\frac{\partial}
	{\partial q}
	RM^{\mathrm{old}}(q)
	=
	c_0SCR_0
	\sum_{t=1}^{T}
	tD_tq^{t-1}
	\geq0,
	\]
	
	and
	
	\[
	\frac{\partial^2}
	{\partial q^2}
	RM^{\mathrm{old}}(q)
	=
	c_0SCR_0
	\sum_{t=2}^{T}
	t(t-1)D_tq^{t-2}
	\geq0.
	\]
	
	Similarly,
	
	\[
	\frac{\partial}
	{\partial q}
	RM^{\mathrm{new}}(q)
	=
	c_1SCR_0
	\sum_{t=1}^{T}
	t\lambda_tD_tq^{t-1}
	\geq0,
	\]
	
	and
	
	\[
	\frac{\partial^2}
	{\partial q^2}
	RM^{\mathrm{new}}(q)
	=
	c_1SCR_0
	\sum_{t=2}^{T}
	t(t-1)\lambda_tD_tq^{t-2}
	\geq0.
	\]
	
	Because \(SCR_0>0\), \(D_t>0\) and \(\lambda_t>0\), the \(t=2\) term makes both functions strictly convex whenever \(T\geq2\).
\end{proof}

The convexity result yields the following comparison.

\begin{corollary}[Jensen effect]
	\label{JensenEffectCor}
	
	For any random \(Q\in[0,1]\),
	
	\begin{equation}
		\begin{aligned}
			\mathbb{E}
			\left[
			RM^{\mathrm{old}}(Q)
			\right]
			&\geq
			RM^{\mathrm{old}}
			\left(
			\mathbb{E}[Q]
			\right),
			\\
			\mathbb{E}
			\left[
			RM^{\mathrm{new}}(Q)
			\right]
			&\geq
			RM^{\mathrm{new}}
			\left(
			\mathbb{E}[Q]
			\right).
		\end{aligned}
		\label{JensenEffect}
	\end{equation}
	
	If \(T\geq2\) and \(Q\) is non-degenerate, both inequalities are strict.
\end{corollary}

Within the reduced-form model, replacing uncertain persistence by its mean therefore understates the scenario-average Risk Margin. The effect is generated by the genuinely multi-period terms \(t\geq2\), since the contributions at \(t=0\) and \(t=1\) are respectively constant and linear in \(Q\). Actuarially, the comparison isolates the capital-cost effect of persistence uncertainty that is omitted when the run-off is represented solely by a mean decay rate.

\subsection{Persistence uncertainty under convex order}
\label{PersistenceConvexOrder}

The Jensen comparison extends to persistence distributions with the same mean.

\begin{proposition}[Persistence uncertainty under convex order]
	\label{PersistenceConvexOrderProp}
	
	Let \(Q_A,Q_B\in[0,1]\) satisfy
	
	\[
	Q_A
	\leq_{\mathrm{cx}}
	Q_B.
	\]
	
	Then
	
	\begin{align*}
		\mathbb{E}
		\left[
		RM^{\mathrm{old}}(Q_A)
		\right]
		&\leq
		\mathbb{E}
		\left[
		RM^{\mathrm{old}}(Q_B)
		\right],
		\\
		\mathbb{E}
		\left[
		RM^{\mathrm{new}}(Q_A)
		\right]
		&\leq
		\mathbb{E}
		\left[
		RM^{\mathrm{new}}(Q_B)
		\right].
	\end{align*}
\end{proposition}

\begin{proof}
	Both Risk Margin functions are convex by Proposition \ref{PersistenceConvexityProp}. The result follows from the defining property of convex order; see Shaked and Shanthikumar (2007)\cite{shaked2007stochastic}.
\end{proof}

At a fixed mean persistence, greater uncertainty in the convex-order sense therefore increases the scenario-average Risk Margin under both calibrations. Expected run-off persistence and uncertainty around that persistence consequently represent distinct dimensions of the projected capital profile.

\subsection{Jensen gap bounds}
\label{JensenGapBounds}

Define

\begin{equation}
	\begin{aligned}
		J_{\mathrm{old}}
		&=
		\mathbb{E}
		\left[
		RM^{\mathrm{old}}(Q)
		\right]
		-
		RM^{\mathrm{old}}
		\left(
		\mathbb{E}[Q]
		\right),
		\\
		J_{\mathrm{new}}
		&=
		\mathbb{E}
		\left[
		RM^{\mathrm{new}}(Q)
		\right]
		-
		RM^{\mathrm{new}}
		\left(
		\mathbb{E}[Q]
		\right).
	\end{aligned}
	\label{JensenGaps}
\end{equation}

For \(t\geq2\), set

\[
\delta_t
=
\mathbb{E}[Q^t]
-
\left(\mathbb{E}[Q]\right)^t.
\]

Convexity of \(q\mapsto q^t\) gives \(\delta_t\geq0\), with strict inequality for non-degenerate \(Q\). Hence

\begin{equation}
	\begin{aligned}
		J_{\mathrm{old}}
		&=
		c_0SCR_0
		\sum_{t=2}^{T}
		D_t\delta_t,
		\\
		J_{\mathrm{new}}
		&=
		c_1SCR_0
		\sum_{t=2}^{T}
		\lambda_tD_t\delta_t.
	\end{aligned}
	\label{JensenGapSums}
\end{equation}

When \(J_{\mathrm{old}}>0\), define

\[
\omega_t
=
\frac{
	D_t\delta_t
}{
	\sum_{s=2}^{T}D_s\delta_s
},
\qquad
t=2,\ldots,T.
\]

Then

\[
\omega_t\geq0,
\qquad
\sum_{t=2}^{T}\omega_t=1.
\]

\begin{proposition}[Relative attenuation of the Jensen gap]
	\label{JensenGapAttenuationProp}
	
	If \(T\geq2\) and \(J_{\mathrm{old}}>0\), then
	
	\begin{equation}
		\begin{aligned}
			\frac{
				J_{\mathrm{new}}
			}{
				J_{\mathrm{old}}
			}
			&=
			\kappa
			\sum_{t=2}^{T}
			\omega_t\lambda_t,
			\\[1mm]
			\kappa\lambda_T
			&\leq
			\frac{
				J_{\mathrm{new}}
			}{
				J_{\mathrm{old}}
			}
			\leq
			\kappa\lambda_2.
		\end{aligned}
		\label{JensenGapRatioBounds}
	\end{equation}
\end{proposition}

\begin{proof}
	From \eqref{JensenGapSums},
	
	\begin{align*}
		\frac{J_{\mathrm{new}}}
		{J_{\mathrm{old}}}
		&=
		\frac{c_1}{c_0}
		\frac{
			\sum_{t=2}^{T}
			\lambda_tD_t\delta_t
		}{
			\sum_{t=2}^{T}
			D_t\delta_t
		}
		\\
		&=
		\kappa
		\sum_{t=2}^{T}
		\omega_t\lambda_t.
	\end{align*}
	
	Since \((\omega_t)\) is a probability distribution and \(\lambda_t\) is non-increasing,
	
	\[
	\lambda_T
	\leq
	\sum_{t=2}^{T}
	\omega_t\lambda_t
	\leq
	\lambda_2.
	\]
	
	Multiplication by \(\kappa\) gives the stated bounds.
\end{proof}

The upper bound is

\[
\kappa\lambda_2
=
\frac{19}{24}(0.96)^2
=
0.7296.
\]

Thus, whenever the Jensen gap is positive, the gap under the revised calibration is at most \(72.96\%\) of its counterpart under the previous calibration. The attenuation combines the lower cost-of-capital rate with the reduced regulatory weights assigned to the multi-period terms through which persistence uncertainty affects the Risk Margin. The ratio compares absolute Jensen gaps measured in Risk Margin units; normalization by the corresponding Risk Margins would define a different quantity.

For \(T\leq1\), or for degenerate \(Q\),

\[
J_{\mathrm{old}}
=
J_{\mathrm{new}}
=
0,
\]

so that the ratio in Proposition \ref{JensenGapAttenuationProp} is not defined.

\subsection{Structural and transitory persistence}
\label{StructuralTransitoryPersistence}

The path-level specification uses a single persistence factor throughout the run-off and therefore represents structural persistence. A contrasting specification uses independently renewed period-specific factors and represents transitory persistence. Define

\begin{equation}
	\begin{aligned}
		SCR_t^{\mathrm{str}}
		&=
		SCR_0Q^t,
		\\
		SCR_t^{\mathrm{tr}}
		&=
		SCR_0
		\prod_{j=1}^{t}Q_j,
	\end{aligned}
	\qquad
	t=0,\ldots,T,
	\label{StructuralTransitorySCR}
\end{equation}

where \(Q_1,Q_2,\ldots\) are independent and identically distributed copies of \(Q\), with the usual empty-product convention at \(t=0\).

The expected capital profiles are

\[
\begin{aligned}
	\mathbb{E}
	\left[
	SCR_t^{\mathrm{str}}
	\right]
	&=
	SCR_0
	\mathbb{E}[Q^t],
	\\
	\mathbb{E}
	\left[
	SCR_t^{\mathrm{tr}}
	\right]
	&=
	SCR_0
	\left(
	\mathbb{E}[Q]
	\right)^t.
\end{aligned}
\]

Let
\(RM_{\mathrm{str}}^{\mathrm{old}}\) and
\(RM_{\mathrm{str}}^{\mathrm{new}}\)
denote the Risk Margins obtained from the structural capital path, and let
\(RM_{\mathrm{tr}}^{\mathrm{old}}\) and
\(RM_{\mathrm{tr}}^{\mathrm{new}}\)
denote the corresponding quantities under transitory persistence.

\begin{proposition}[Structural versus transitory persistence]
	\label{StructuralTransitoryProp}
	
	For every \(t\geq0\),
	
	\[
	\mathbb{E}
	\left[
	SCR_t^{\mathrm{str}}
	\right]
	\geq
	\mathbb{E}
	\left[
	SCR_t^{\mathrm{tr}}
	\right].
	\]
	
	Consequently,
	
	\[
	\mathbb{E}
	\left[
	RM_{\mathrm{str}}^{\mathrm{old}}
	\right]
	\geq
	\mathbb{E}
	\left[
	RM_{\mathrm{tr}}^{\mathrm{old}}
	\right]
	\]
	
	and
	
	\[
	\mathbb{E}
	\left[
	RM_{\mathrm{str}}^{\mathrm{new}}
	\right]
	\geq
	\mathbb{E}
	\left[
	RM_{\mathrm{tr}}^{\mathrm{new}}
	\right].
	\]
	
	If \(T\geq2\) and \(Q\) is non-degenerate, both Risk Margin inequalities are strict.
\end{proposition}

\begin{proof}
	Equality holds for \(t=0\) and \(t=1\). For every integer \(t\geq2\), the mapping
	
	\[
	q\mapsto q^t
	\]
	
	is strictly convex on \([0,1]\). Jensen's inequality therefore gives
	
	\[
	\mathbb{E}[Q^t]
	\geq
	\left(
	\mathbb{E}[Q]
	\right)^t,
	\]
	
	with strict inequality when \(Q\) is non-degenerate. Multiplication by \(SCR_0\) gives the capital-profile comparison. Since the discount factors and regulatory weights are strictly positive, summation over \(t\) yields the corresponding Risk Margin inequalities.
\end{proof}

Structural persistence therefore produces a higher expected future capital profile and a larger scenario-average Risk Margin than independently renewed persistence factors sharing the same one-period marginal distribution. The difference reflects dependence across the run-off. The proposition establishes an ordering of expectations; pathwise and stochastic comparisons require additional assumptions.
\section{Numerical analysis}
\label{NumericalAnalysis}

The numerical analysis combines run-off profiles documented in the actuarial literature with controlled experiments designed to isolate the mechanisms established above. The profiles reported by England et al. (2019)\cite{england2019lifetime} and Dreksler et al. (2015)\cite{dreksler2015solvency} serve as documented numerical benchmarks, while synthetic constructions vary timing, dispersion, persistence, horizon and interest-rate exposure independently. Unless otherwise stated, the calculations use \(c_0=6\%\), \(c_1=4.75\%\), and the time-dependent factor defined in \eqref{LambdaFactor}.

\subsection{Controlled run-off benchmarks}
\label{ControlledRunoffBenchmarks}

Consider first the limiting case in which the entire discounted SCR contribution is concentrated at a single run-off date \(t\). Proposition \ref{ExactRepresentationProp} then gives

\[
\Delta_t
=
1-\kappa\lambda_t.
\]

Figure \ref{Fig-1-SingleDateReduction} reports the corresponding reform effect over a 40-year horizon. The vertical reference line marks \(t=17\), the first integer date at which the \(50\%\) floor becomes binding.

\begin{figure}[htbp]
	\centering
	\includegraphics[width=0.82\textwidth]{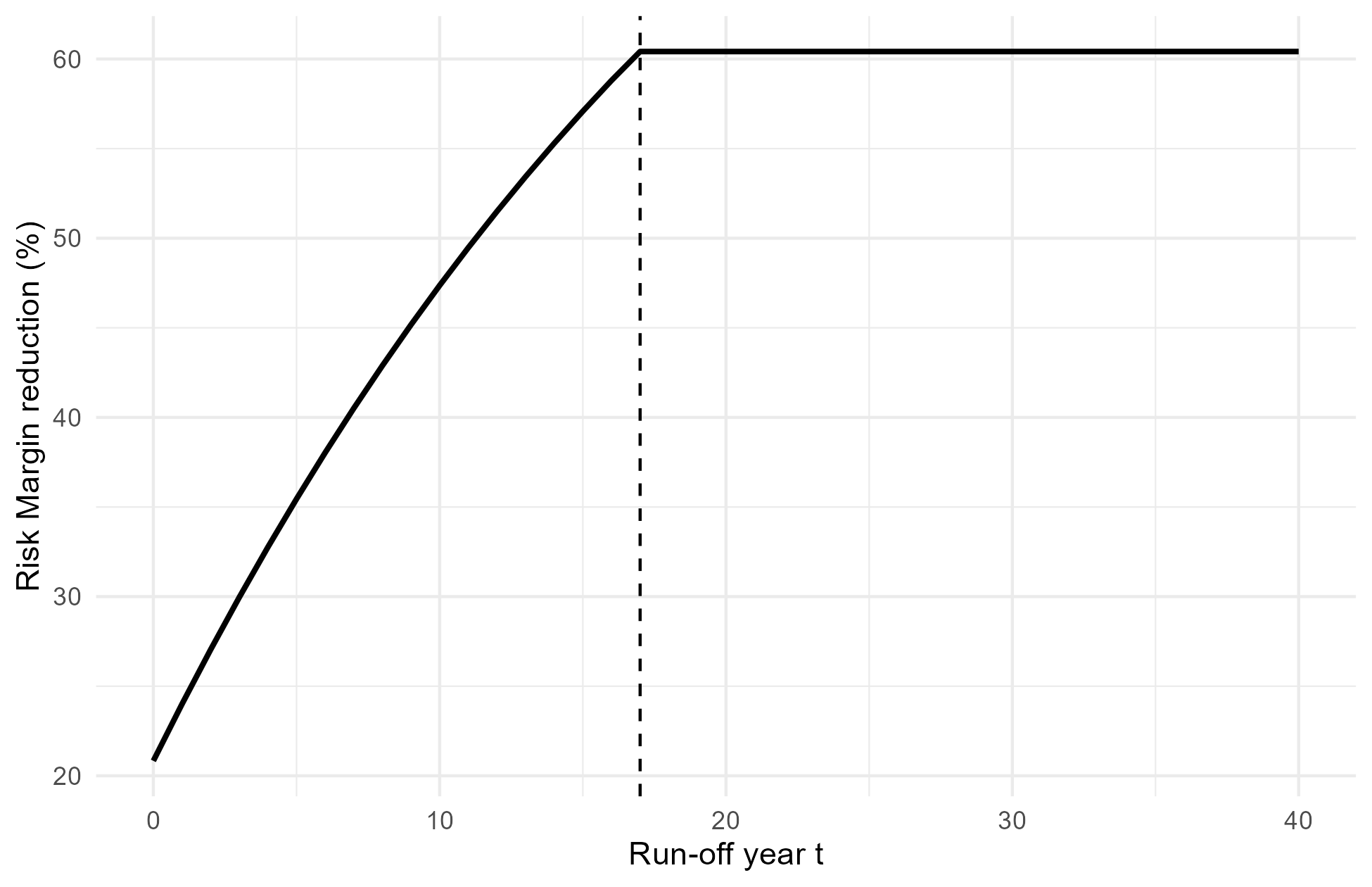}
	\caption{Percentage Risk Margin reduction for a single-date discounted SCR contribution.}
	\label{Fig-1-SingleDateReduction}
\end{figure}

At \(t=0\), \(\lambda_0=1\), so the only change is the reduction in the prescribed cost-of-capital rate and

\[
\Delta_0
=
1-\kappa
\simeq
20.83\%.
\]

As the capital contribution is shifted to later dates, the time-dependent adjustment produces an additional reduction. Once the floor becomes binding,

\[
\Delta_t
=
1-\frac{19}{48}
\simeq
60.42\%,
\qquad t\geq17.
\]

The single-date experiment therefore attains the two universal extremes established in Corollary \ref{UniversalBoundsCor}.

A second experiment isolates temporal dispersion while holding the mean run-off time fixed. Consider two discounted SCR distributions supported on \(\{0,\ldots,20\}\), both satisfying

\[
\mathbb{E}[\tau]=10.
\]

The first places all mass at \(t=10\). The second assigns probability \(1/2\) to each endpoint \(t=0\) and \(t=20\). Their standard deviations are \(0\) and \(10\), respectively.

The concentrated profile produces a Risk Margin reduction of \(47.37\%\), compared with \(40.63\%\) for the endpoint distribution. The difference of approximately \(6.74\) percentage points arises despite identical horizons and mean run-off times.

For \(T=20\) and \(\mu=10\), Theorem \ref{SharpMeanBoundsThm} gives the sharp interval

\[
40.625\%
\leq
\Delta
\leq
47.367\%.
\]

Both endpoints are attained by the two distributions above. The experiment therefore illustrates simultaneously the convex-order result in Proposition \ref{ConvexOrderProp} and the sharpness of the fixed-mean bounds.

\subsection{Published run-off profiles}
\label{PublishedRunoffProfiles}

We next consider the capital run-off profiles reported by England et al. (2019)\cite{england2019lifetime}. Their numerical analysis constructs three future capital profiles based on the Best Estimate, the standard deviation of future claims development results, and the \(99.5\%\) Value-at-Risk. These alternatives are derived from the same underlying reserving problem and therefore provide a useful setting for comparing different temporal patterns of capital release.

The example uses a \(6\%\) cost-of-capital rate and a flat \(3\%\) discount rate. Before introducing the revised calibration, the Risk Margins implied by the published capital profiles are reconstructed as \(818{,}047.45\), \(822{,}320.99\), and \(810{,}816.20\) for the Best Estimate, standard-deviation and VaR bases, respectively. These reproduce the reported values \(818{,}047\), \(822{,}321\), and \(810{,}816\) to rounding accuracy.

Figure \ref{Fig-2-EnglandProfiles} displays the three capital profiles after normalising the opening requirement to 100.

\begin{figure}[htbp]
	\centering
	\includegraphics[width=0.84\textwidth]{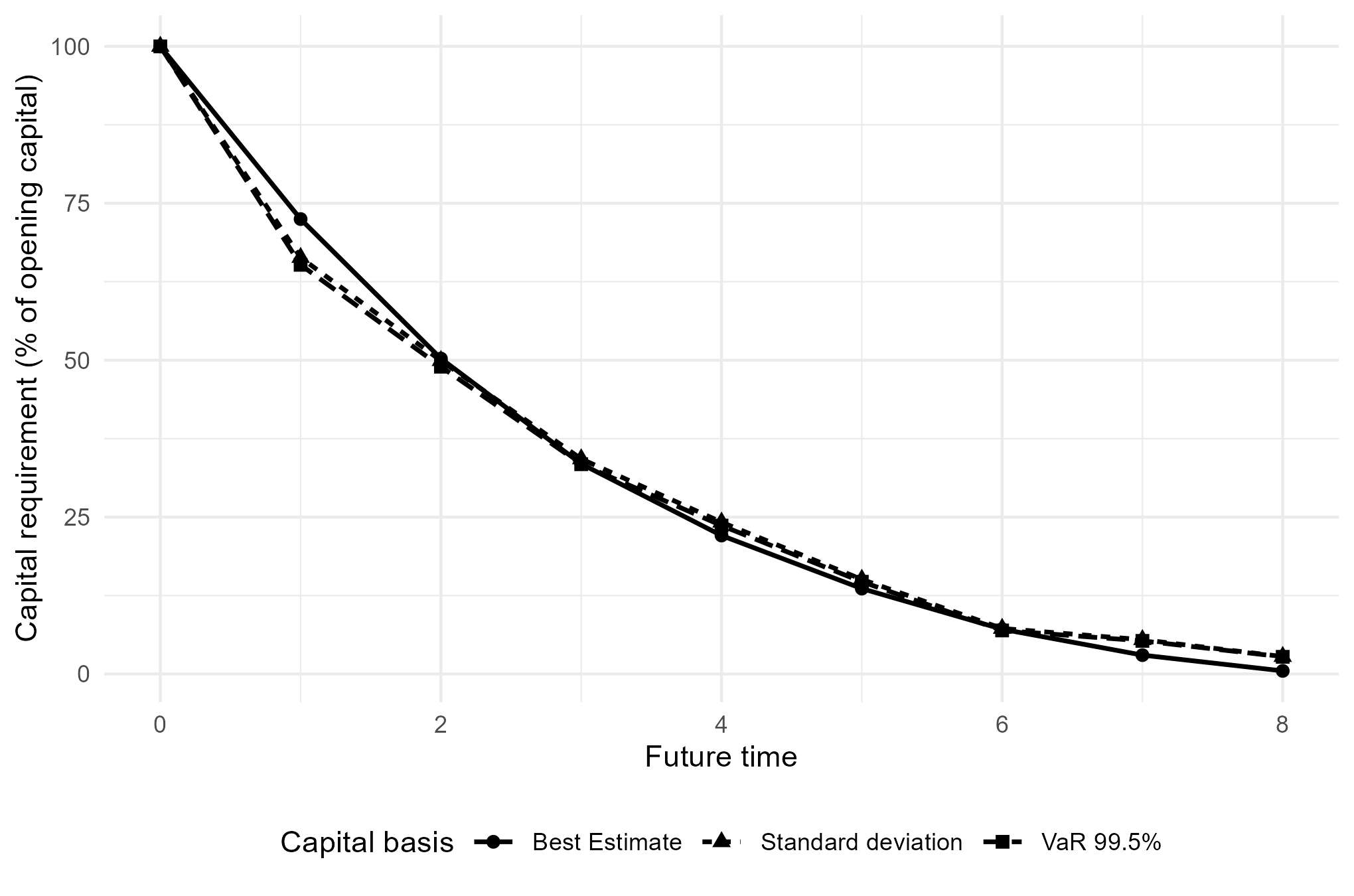}
	\caption{Published capital run-off profiles normalised to 100 at the initial date.}
	\label{Fig-2-EnglandProfiles}
\end{figure}

The standard-deviation and VaR profiles are particularly close. England et al. (2019)\cite{england2019lifetime} relate this behaviour to the approximate proportionality between a high quantile and the standard deviation when the aggregate claims development distribution is sufficiently close to normality. After normalisation at the initial date, this produces similar temporal profiles.

The proximity of the Best Estimate profile has a different interpretation. The run-off of expected remaining liabilities can differ from both the standard deviation and a high quantile of future reserve risk. Its similarity to the other profiles is therefore a feature of this numerical example rather than a structural relationship.

For each profile, let \(\sigma_\tau=\sqrt{\operatorname{Var}(\tau)}\) denote the standard deviation of the discounted SCR run-off time. Table \ref{Tab-1-PublishedProfiles} applies the 2027 calibration to the three England et al. profiles and to the SCR run-off reported by Dreksler et al. (2015)\cite{dreksler2015solvency}.

\begin{table}[htbp]
	\centering
	\small
	\caption{Reform effects for published run-off profiles.}
	\label{Tab-1-PublishedProfiles}
	\begin{tabular}{lrrrrrr}
		\toprule
		Profile & \(T\) & \(RM^{\mathrm{old}}\) & \(RM^{\mathrm{new}}\) &
		Reduction (\%) & \(\mu\) & \(\sigma_\tau\) \\
		\midrule
		England et al. Best Estimate
		& 8 & 818047.45 & 609031.44 & 25.55 & 1.560 & 1.665 \\
		England et al. Standard deviation
		& 8 & 822320.99 & 609365.16 & 25.90 & 1.685 & 1.816 \\
		England et al. VaR 99.5\%
		& 8 & 810816.20 & 601249.15 & 25.85 & 1.668 & 1.811 \\
		Dreksler et al. illustrative profile
		& 7 & 1390.00 & 1038.80 & 25.27 & 1.476 & 1.798 \\
		\bottomrule
	\end{tabular}
\end{table}

The Dreksler et al. (2015)\cite{dreksler2015solvency} profile is explicitly presented by the authors as illustrative dummy data. It is used here solely as an additional published benchmark for the reform mechanics. Because the published table contains rounded discounted cost-of-capital contributions, the revised Risk Margin reconstructed from those quantities is approximate.

Comparisons between the previous and revised calibrations are conducted within each published profile, since the absolute Risk Margin levels of the England and Dreksler examples are expressed on different scales.

The four reductions lie between approximately \(25\%\) and \(26\%\), only moderately above the \(20.83\%\) reduction generated by the cost-of-capital change alone. Their relatively short capital run-offs leave most discounted SCR contributions in the early years, where the lambda attenuation remains limited.

The fixed-mean and horizon bounds are correspondingly tight. For the England Best Estimate profile, \(T=8\) and \(\mu=1.560\) imply

\[
25.13\%
\leq
\Delta
\leq
25.70\%,
\]

compared with an actual reduction of \(25.55\%\). The corresponding intervals are approximately

\[
25.48\%
\leq
\Delta
\leq
26.08\%
\]

for the standard-deviation profile and

\[
25.43\%
\leq
\Delta
\leq
26.03\%
\]

for the VaR profile. The mean discounted-SCR run-off time and terminal horizon therefore constrain the reform effect closely in these short-run-off examples.

\subsection{Best Estimate and alternative capital profiles}
\label{BestEstimateAlternativeProfiles}

The England et al. (2019)\cite{england2019lifetime} profiles also allow the proportional Best Estimate approximation to be examined directly. EIOPA's Method 2 approximates future SCRs proportionally to the projected Best Estimate \cite{eiopa2026valuation}. The published Best Estimate sequence is therefore used as the liability profile, while the Best Estimate, standard-deviation and VaR series are successively treated as alternative capital paths.

Departures from proportionality are measured by the capital-intensity sequence

\[
h_t
=
\frac{SCR_t}{BE_t}.
\]

Table \ref{Tab-2-CapitalProfileComparison} compares the full reform ratio with its proportional Best Estimate counterpart. Define the revised Risk Margin level error as

\[
E_{\mathrm{RM}}
=
RM_{\mathrm{full}}^{\mathrm{new}}
-
RM_{\mathrm{M2}}^{\mathrm{new}}.
\]

\begin{table}[htbp]
	\centering
	\small
	\caption{Best Estimate approximation across capital profiles.}
	\label{Tab-2-CapitalProfileComparison}
	\begin{tabular}{lrrrrrr}
		\toprule
		Capital basis &
		\(R_{\mathrm{full}}\) &
		\(R_{\mathrm{M2}}\) &
		\(\Delta_{\mathrm{full}}\) (\%) &
		\(\Delta_{\mathrm{M2}}\) (\%) &
		\(\operatorname{Cov}_{\pi}(\lambda_\tau,h_\tau)\) &
		\(E_{\mathrm{RM}}\) \\
		\midrule
		Best Estimate
		& 0.744494 & 0.744494 & 25.551 & 25.551
		& \(-1.60\times10^{-9}\) & 0.03 \\
		Standard deviation
		& 0.741031 & 0.744494 & 25.897 & 25.551
		& \(-1.23\times10^{-3}\) & 333.75 \\
		VaR 99.5\%
		& 0.741536 & 0.744494 & 25.846 & 25.551
		& \(-1.04\times10^{-3}\) & -7782.25 \\
		\bottomrule
	\end{tabular}
\end{table}

When the Best Estimate itself is used as the capital basis, proportionality is recovered essentially exactly, with the negligible residual discrepancy arising from rounding in the published inputs.

For both the standard-deviation and VaR capital profiles,

\[
\operatorname{Cov}_{\pi}
(\lambda_\tau,h_\tau)
<
0.
\]

Proposition \ref{CapitalIntensityProp} therefore implies

\[
R_{\mathrm{full}}
<
R_{\mathrm{M2}},
\]

so that the proportional Best Estimate approximation slightly understates the percentage reduction generated by the reform. The Method 2 reduction is \(25.55\%\), compared with \(25.90\%\) under the standard-deviation profile and \(25.85\%\) under the VaR profile.

The Risk Margin level error behaves differently. For the standard-deviation profile,

\[
E_{\mathrm{RM}}
\simeq
333.75,
\]

whereas for the VaR profile,

\[
E_{\mathrm{RM}}
\simeq
-7{,}782.25.
\]

Relative to the corresponding full revised Risk Margins, these errors are approximately \(0.055\%\) and \(-1.294\%\), respectively. The Best Estimate case is effectively exact at the reported precision.

The change in sign illustrates the distinction established analytically in Section \ref{CapitalIntensityMethod2}. The covariance in Proposition \ref{CapitalIntensityProp} determines the difference between the relative reform ratios, whereas the approximation error in the Risk Margin level also depends on the average capital-intensity level. The sign of the absolute valuation error is therefore determined jointly by the timing interaction and the capital-intensity level.

The standard-deviation and VaR capital-intensity sequences vary non-monotonically over the run-off. The exact covariance identity therefore provides the applicable characterization beyond the monotonicity-based sufficient conditions in Corollary \ref{MonotoneCapitalIntensityCor}.

\subsection{Long-horizon and interest-rate scenarios}
\label{LongHorizonRates}

The published profiles have relatively short horizons and jointly determine persistence and maturity. The geometric benchmark introduced in Section \ref{GeometricRunoff},

\[
SCR_t
=
SCR_0q^t,
\]

provides a controlled setting in which these features can be varied independently.

The annual effective risk-free rate is fixed at \(3\%\), \(q\) varies between \(0.05\) and \(0.99\), and the horizons are set to

\[
T=10,\qquad T=20,\qquad T=40.
\]

Figure \ref{Fig-3-GeometricRunoff} reports the resulting percentage reductions.

\begin{figure}[htbp]
	\centering
	\includegraphics[width=0.82\textwidth]{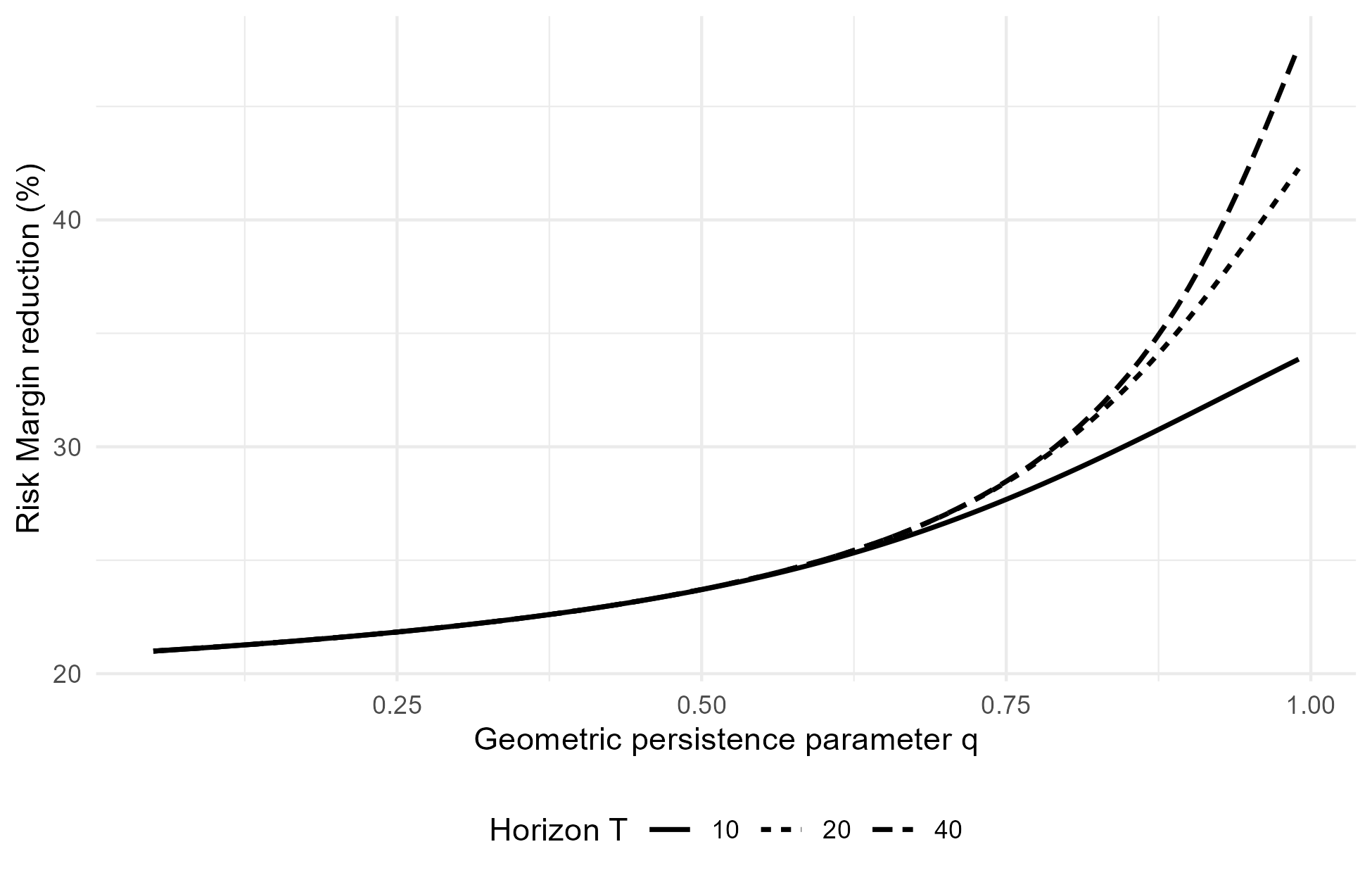}
	\caption{Percentage Risk Margin reduction under geometric SCR run-off.}
	\label{Fig-3-GeometricRunoff}
\end{figure}

The numerical pattern follows Proposition \ref{PersistenceMonotonicityProp}. When \(q\) is small, capital releases rapidly and the reduction remains close to the effect of the cost-of-capital change alone. As \(q\) approaches one, a larger share of discounted capital remains at later dates and receives lower lambda factors.

The effect of extending the horizon is correspondingly limited for rapidly decaying profiles and becomes material when persistence is sufficiently high. At \(q=0.99\), the reductions are approximately \(33.86\%\), \(42.27\%\), and \(47.65\%\) for \(T=10\), \(20\), and \(40\), respectively. Large reform effects therefore arise from the combination of a long horizon and persistent future capital requirements.

The direct interest-rate mechanism derived in Section \ref{InterestRateSensitivity} is examined using three geometric profiles,

\[
(T,q)
=
(10,0.65),\qquad
(25,0.85),\qquad
(50,0.95).
\]

The labels short, medium and long distinguish these three constructed run-off patterns.

Starting from the flat \(3\%\) baseline curve, parallel continuously compounded shifts ranging from \(-200\) to \(+200\) basis points are applied while holding the projected SCR paths fixed. The experiment therefore isolates the direct discounting channel characterized in Propositions \ref{RateSensitivityReductionProp} and \ref{ReformRatioRateProp}.

Figure \ref{Fig-4-InterestRateSensitivity} reports the resulting reform effects.

\begin{figure}[htbp]
	\centering
	\includegraphics[width=0.84\textwidth]{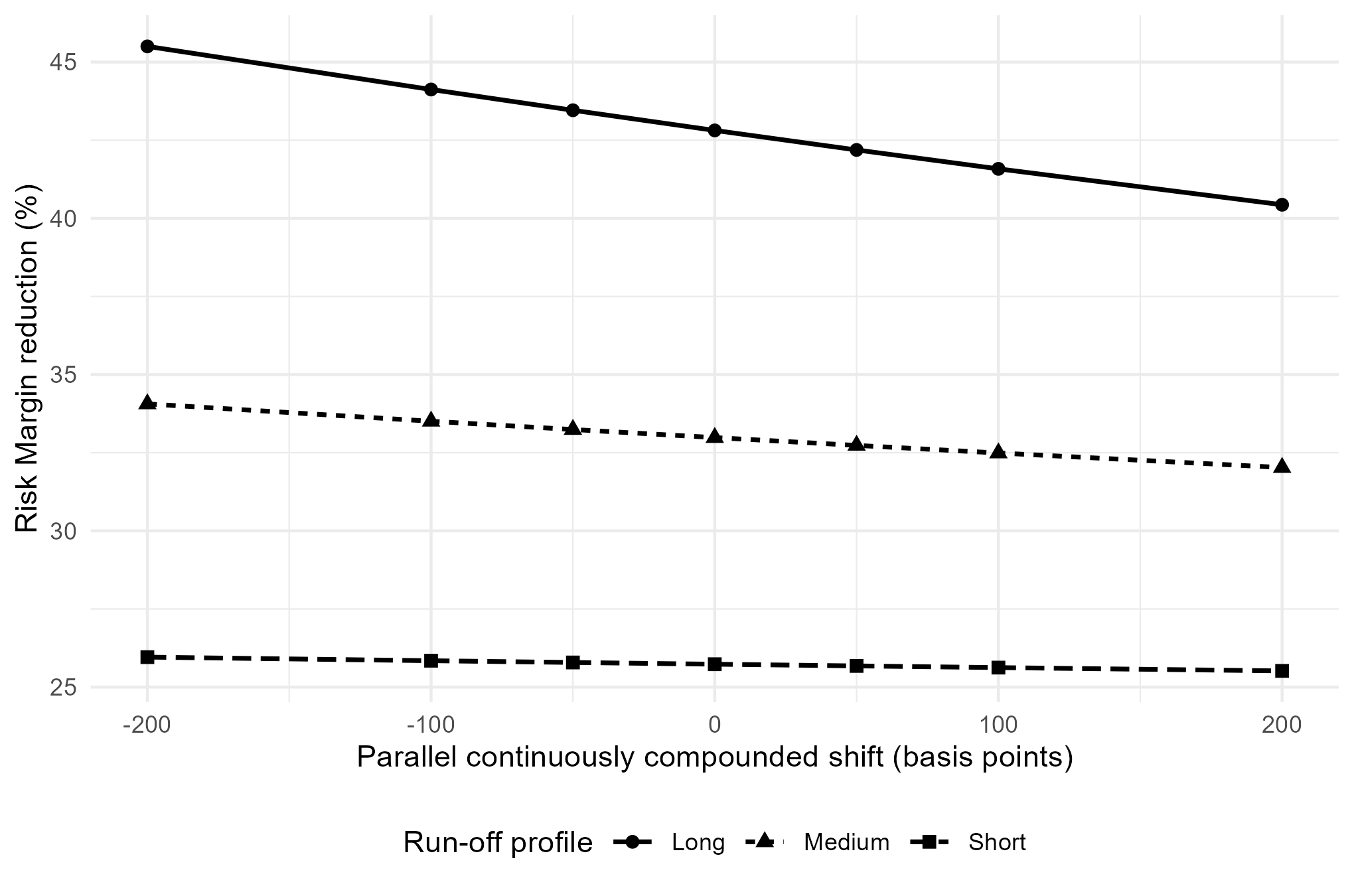}
	\caption{Percentage Risk Margin reduction under parallel continuously compounded rate shifts.}
	\label{Fig-4-InterestRateSensitivity}
\end{figure}

At the baseline curve, the Risk Margin reductions are \(25.73\%\), \(32.98\%\), and \(42.81\%\) for the short, medium and long profiles. The direct semi-elasticities under the previous calibration are

\[
2.641,\qquad
5.545,\qquad
12.036,
\]

compared with

\[
2.492,\qquad
4.787,\qquad
9.815
\]

under the revised calibration. Relative to the previous values, the direct sensitivity reductions are approximately \(5.6\%\), \(13.7\%\), and \(18.5\%\).

Across these three profiles, increasing the continuously compounded shift from \(-200\) to \(+200\) basis points lowers the percentage reduction generated by the reform by approximately \(0.44\), \(2.04\), and \(5.07\) percentage points for the short, medium and long profiles, respectively.

Higher rates reduce the present value of distant capital contributions more strongly. Since those contributions also receive the lowest lambda factors, their declining weight brings the previous and revised calibrations relatively closer. Lower rates increase the relative weight of these distant contributions and strengthen the reform effect. Across the constructed profiles, the magnitude of this direct channel increases with the persistence and horizon of the capital run-off. This mechanism is characterized by Proposition \ref{ReformRatioRateProp} and is consistent with the long-duration numerical behaviour discussed by Waszink (2024)\cite{waszink2024comments}.

\subsection{Stochastic persistence experiments}
\label{StochasticPersistenceExperiments}

The final experiment quantifies the persistence-uncertainty effects developed in Section \ref{StochasticRunoffPersistence}. Let

\[
Q
\sim
\operatorname{Beta}(a,b),
\]

with fixed mean

\[
\mathbb{E}[Q]=0.85.
\]

Set

\[
a=0.85\phi,
\qquad
b=0.15\phi,
\qquad
\phi>0,
\]

where \(\phi\) is a concentration parameter. Varying \(\phi\) changes the dispersion of \(Q\) while preserving its mean. The calculations use \(T=40\), \(SCR_0=1\), and a flat annual effective risk-free rate of \(3\%\). The scenario averages are evaluated analytically from the moments of the Beta distribution.

Table \ref{Tab-3-StochasticPersistence} reports the scenario-average Risk Margins under structural persistence together with the corresponding Jensen gaps.

\begin{table}[htbp]
	\centering
	\small
	\caption{Scenario-average Risk Margins and Jensen gaps under stochastic persistence.}
	\label{Tab-3-StochasticPersistence}
	\begin{tabular}{rrrrrrr}
		\toprule
		\(\phi\) &
		\(\operatorname{Var}(Q)\) &
		\(\mathbb{E}[RM_{\mathrm{old}}^{\mathrm{str}}]\) &
		\(\mathbb{E}[RM_{\mathrm{new}}^{\mathrm{str}}]\) &
		\(J_{\mathrm{old}}\) &
		\(J_{\mathrm{new}}\) &
		\(J_{\mathrm{new}}/J_{\mathrm{old}}\) \\
		\midrule
		100 & 0.001262 & 0.3463 & 0.2289 & 0.0131 & 0.0062 & 0.4714 \\
		30  & 0.004113 & 0.3776 & 0.2434 & 0.0444 & 0.0207 & 0.4661 \\
		15  & 0.007969 & 0.4215 & 0.2635 & 0.0883 & 0.0408 & 0.4616 \\
		8   & 0.014167 & 0.4907 & 0.2948 & 0.1575 & 0.0721 & 0.4575 \\
		4   & 0.025500 & 0.6054 & 0.3463 & 0.2722 & 0.1236 & 0.4541 \\
		\bottomrule
	\end{tabular}
\end{table}

The Jensen gaps are positive in every scenario, in accordance with Corollary \ref{JensenEffectCor}. Within this model, replacing uncertain path-level persistence by its mean therefore understates the scenario-average Risk Margin.

Under transitory persistence,

\[
\mathbb{E}
\left[
SCR_t^{\mathrm{tr}}
\right]
=
SCR_0
\left(
\mathbb{E}[Q]
\right)^t.
\]

Because the mean of \(Q\) is fixed across the five specifications, the expected transitory Risk Margins are identical across cases, at approximately \(0.3332\) under the previous calibration and \(0.2227\) under the revised calibration. The corresponding structural values are strictly larger, as established in Proposition \ref{StructuralTransitoryProp}.

Within the selected Beta family, the increase in \(\operatorname{Var}(Q)\) is accompanied by a substantial increase in both Jensen gaps. The previous-calibration gap rises from approximately \(0.0131\) at

\[
\operatorname{Var}(Q)=0.001262
\]

to \(0.2722\) at

\[
\operatorname{Var}(Q)=0.025500.
\]

Over the same range, the revised-calibration gap increases from \(0.0062\) to \(0.1236\).

This monotonic numerical pattern is specific to the selected parametric family. The general comparison established in Proposition \ref{PersistenceConvexOrderProp} is formulated in terms of convex order.

The revised calibration substantially attenuates the Jensen gap. Across the five specifications,

\[
0.4541
\leq
\frac{J_{\mathrm{new}}}{J_{\mathrm{old}}}
\leq
0.4714,
\]

within the analytical interval from Proposition \ref{JensenGapAttenuationProp},

\[
0.3958
\leq
\frac{J_{\mathrm{new}}}{J_{\mathrm{old}}}
\leq
0.7296.
\]

Persistence uncertainty therefore continues to increase the scenario-average cost of future capital, while the combined effect of the lower cost-of-capital rate and the lambda weighting reduces the associated convexity contribution substantially.

Taken together, the numerical results indicate that neither the opening SCR nor the terminal projection horizon alone identifies the effect of the reform. Its magnitude depends on the temporal allocation and dispersion of discounted capital, the relationship between the SCR and Best Estimate run-off profiles, and the structural persistence of future capital requirements. These features are therefore central to actuarial projection, model validation and the interpretation of the 2027 Risk Margin.
\section{Discussion}
\label{Discussion}

Under both the previous and revised cost-of-capital calibrations, the timing of future capital affects the Risk Margin through discounting. The 2027 reform adds a second temporal channel by applying different regulatory weights according to the position of each capital contribution in the run-off. The reduction in the prescribed cost-of-capital rate remains uniform across the projected SCR path, whereas the lambda factor creates additional portfolio heterogeneity. Portfolios with comparable opening SCRs or Risk Margins can consequently experience different reform effects when their future capital requirements are released at different speeds.

Mean discounted-SCR run-off time must be complemented by information on temporal dispersion. At a fixed mean, different allocations of capital across the projection horizon can generate materially different reductions because the lambda sequence is nonlinear. The sharp bounds conditional on the horizon and mean quantify the remaining range of possible reform effects and thereby measure the information lost when the full capital profile is replaced by a one-dimensional timing summary. This range becomes particularly relevant for profiles extending towards the region in which the regulatory floor becomes binding.

The \(50\%\) floor is a lower bound on the lambda factor. Combined with the reduction in the prescribed cost-of-capital rate, it produces a maximum mechanical Risk Margin reduction of approximately \(60.42\%\). Attainment requires the entire discounted SCR mass to be located at dates for which the floor applies. The value is therefore a sharp theoretical limit associated with an extremal capital profile. The published run-off profiles considered in the numerical analysis remain well below this limit because most of their discounted capital is concentrated in the early years.

The interaction with simplified SCR projection introduces a separate model-validation issue. Under the proportional Best Estimate approximation, the difference between the full and approximated reform ratios is determined exactly by the covariance between capital intensity and the lambda sequence. A zero covariance preserves the reform ratio, with constant capital intensity representing one sufficient configuration. The approximation error in the Risk Margin level follows a distinct decomposition that also depends on the average capital-intensity level. Relative reform effects and absolute valuation errors are therefore governed by different combinations of timing and scale.

An approximation can reproduce the aggregate level of projected capital while allocating that capital differently across the run-off. Under the revised calibration, such timing discrepancies enter the Risk Margin through the lambda weighting. The England et al. (2019)\cite{england2019lifetime} profiles illustrate this mechanism in a setting where the alternative capital run-offs remain relatively close. Their numerical differences are modest, while the covariance representation identifies the channel through which larger departures between liability and capital run-offs would affect the reform ratio. A corresponding validation of proportional projection therefore examines the temporal profile of capital intensity in addition to aggregate approximation errors.

This model-validation requirement also reinforces the separation between the underlying evolution of risk and the regulatory attenuation applied after the SCR projection. The revised EIOPA guidance requires projected SCRs to exclude the time dependency already represented by the lambda factor \cite{eiopa2026valuation}. Under simplified projection methods, embedding an additional regulatory-style attenuation within the projected capital path would modify its temporal structure before application of the prescribed lambda sequence and would therefore duplicate the intended adjustment.

The interest-rate results provide a further implication of temporal reweighting. Conditional on a fixed projected SCR path and parallel shifts in continuously compounded discount rates, the revised Risk Margin has a weakly lower direct semi-elasticity than the previous calibration. The inequality is strict when the discounted SCR support includes dates associated with distinct lambda factors. Across the controlled profiles, the difference is limited for rapidly declining capital paths and becomes larger as capital remains material further into the run-off.

This comparison characterizes the direct discounting channel. The total response of the Risk Margin and the wider balance sheet additionally reflects the sensitivity of projected SCRs to interest rates. As shown in Section \ref{InterestRateSensitivity}, the ordering of total semi-elasticities depends on the term structure of this capital response and is therefore portfolio and capital-model dependent. This separation between direct and capital-mediated effects provides the appropriate scope for interpreting the lower sensitivity associated with lambda reweighting and is consistent with the broader sensitivity analysis discussed by Waszink (2024)\cite{waszink2024comments}.

The stochastic persistence analysis extends the timing perspective to projection uncertainty. Within the reduced-form model, the Risk Margin is convex in the path-level persistence parameter. For non-degenerate persistence and a run-off containing contributions at dates \(t\geq2\), replacing uncertain persistence by its mean understates the scenario-average Risk Margin. The expectation over persistence scenarios serves as a diagnostic measure of projection uncertainty, while the regulatory Risk Margin remains the conditional calculation associated with a given projected SCR path.

The comparison between structural and transitory persistence isolates the role of dependence across time. Applying a common path-level factor generates a weakly larger scenario-average accumulation of future capital than independently renewing persistence factors drawn from the same one-period distribution. The difference is strict under non-degenerate multi-period persistence and arises from dependence across the run-off rather than from the marginal distribution of the individual factors. The revised calibration attenuates the associated Jensen effect while preserving the influence of persistence uncertainty.

The numerical evidence has two distinct functions. The England et al. (2019)\cite{england2019lifetime} and Dreksler et al. (2015)\cite{dreksler2015solvency} profiles provide published actuarial benchmarks for relatively short run-offs. The long-horizon, interest-rate and stochastic-persistence experiments are controlled constructions that isolate mechanisms jointly determined in those published examples. The larger effects generated by persistent capital extending into the floor region are therefore comparative-static implications of the regulatory structure. Their translation into portfolio-level magnitudes requires projected SCR profiles reflecting the liabilities, risk composition and capital model of the undertaking concerned.

The analysis deliberately isolates the reform of the Risk Margin formula. Except in the experiments designed to vary a specific input, the projected SCR path before application of the regulatory time-dependent factor and the discounting structure are held common across the previous and revised calculations. The wider 2027 Solvency II review can affect technical provisions, capital requirements and other balance-sheet components through additional channels. The results characterize the revised cost-of-capital rate and lambda weighting conditional on the capital dynamics supplied to the Risk Margin calculation.

For actuarial risk management, the reform increases the importance of the projected distribution of capital across the liability run-off. Model governance consequently encompasses the timing and dispersion of future SCR contributions, the evolution of capital intensity relative to the Best Estimate, and the sensitivity of these features to economic and insurance-risk assumptions. This profile-based assessment complements validation based on the opening SCR or aggregate Risk Margin and provides a direct link between regulatory calibration, capital projection and the management of long-duration insurance obligations.

\section*{Conclusion}

This paper develops an analytical framework for the final 2027 calibration of the Solvency II Risk Margin. Normalizing discounted projected SCR contributions into a distribution over run-off time converts the new-to-old Risk Margin ratio into an expectation of the regulatory lambda factor. This representation provides a common basis for analyzing run-off timing, temporal dispersion, proportional Best Estimate projection, direct interest-rate sensitivity and uncertainty in capital persistence. It yields exact identities, stochastic-ordering results and sharp bounds while separating relative reform effects from approximation errors in the Risk Margin level.

The central conclusion is that the effect of the reform is a property of the entire projected capital run-off. Opening capital and aggregate liability measures provide only partial information once future SCR contributions receive different regulatory weights according to their timing. For undertakings and supervisors, implementation therefore places greater emphasis on the temporal validity of SCR projections and on the separation between underlying risk dynamics and the attenuation already incorporated through the lambda factor.

Future research can extend the framework to jointly stochastic interest rates and capital requirements, richer dependence structures and portfolio-level empirical applications. Embedding these elements in a full balance-sheet model would connect the transparent comparative statics developed here with firm-specific solvency assessment and risk management.

\bibliography{S2-RM-Bib}


\begin{thebibliography}{29}
\ifx \bisbn   \undefined \def \bisbn  #1{ISBN #1}\fi
\ifx \binits  \undefined \def \binits#1{#1}\fi
\ifx \bauthor  \undefined \def \bauthor#1{#1}\fi
\ifx \batitle  \undefined \def \batitle#1{#1}\fi
\ifx \bjtitle  \undefined \def \bjtitle#1{#1}\fi
\ifx \bvolume  \undefined \def \bvolume#1{\textbf{#1}}\fi
\ifx \byear  \undefined \def \byear#1{#1}\fi
\ifx \bissue  \undefined \def \bissue#1{#1}\fi
\ifx \bfpage  \undefined \def \bfpage#1{#1}\fi
\ifx \blpage  \undefined \def \blpage #1{#1}\fi
\ifx \burl  \undefined \def \burl#1{\textsf{#1}}\fi
\ifx \doiurl  \undefined \def \doiurl#1{\url{https://doi.org/#1}}\fi
\ifx \betal  \undefined \def \betal{\textit{et al.}}\fi
\ifx \binstitute  \undefined \def \binstitute#1{#1}\fi
\ifx \binstitutionaled  \undefined \def \binstitutionaled#1{#1}\fi
\ifx \bctitle  \undefined \def \bctitle#1{#1}\fi
\ifx \beditor  \undefined \def \beditor#1{#1}\fi
\ifx \bpublisher  \undefined \def \bpublisher#1{#1}\fi
\ifx \bbtitle  \undefined \def \bbtitle#1{#1}\fi
\ifx \bedition  \undefined \def \bedition#1{#1}\fi
\ifx \bseriesno  \undefined \def \bseriesno#1{#1}\fi
\ifx \blocation  \undefined \def \blocation#1{#1}\fi
\ifx \bsertitle  \undefined \def \bsertitle#1{#1}\fi
\ifx \bsnm \undefined \def \bsnm#1{#1}\fi
\ifx \bsuffix \undefined \def \bsuffix#1{#1}\fi
\ifx \bparticle \undefined \def \bparticle#1{#1}\fi
\ifx \barticle \undefined \def \barticle#1{#1}\fi
\bibcommenthead
\ifx \bconfdate \undefined \def \bconfdate #1{#1}\fi
\ifx \botherref \undefined \def \botherref #1{#1}\fi
\ifx \url \undefined \def \url#1{\textsf{#1}}\fi
\ifx \bchapter \undefined \def \bchapter#1{#1}\fi
\ifx \bbook \undefined \def \bbook#1{#1}\fi
\ifx \bcomment \undefined \def \bcomment#1{#1}\fi
\ifx \oauthor \undefined \def \oauthor#1{#1}\fi
\ifx \citeauthoryear \undefined \def \citeauthoryear#1{#1}\fi
\ifx \endbibitem  \undefined \def \endbibitem {}\fi
\ifx \bconflocation  \undefined \def \bconflocation#1{#1}\fi
\ifx \arxivurl  \undefined \def \arxivurl#1{\textsf{#1}}\fi
\csname PreBibitemsHook\endcsname

\bibitem[\protect\citeauthoryear{Albrecher et~al.}{2022}]{albrecher2022cost}
\begin{barticle}
\bauthor{\bsnm{Albrecher}, \binits{H.}},
\bauthor{\bsnm{Eisele}, \binits{K.-T.}},
\bauthor{\bsnm{Steffensen}, \binits{M.}},
\bauthor{\bsnm{W{\"u}thrich}, \binits{M.V.}}:
\batitle{On the cost-of-capital rate under incomplete market valuation}.
\bjtitle{Journal of Risk and Insurance}
\bvolume{89}(\bissue{4}),
\bfpage{1139}--\blpage{1158}
(\byear{2022})
\doiurl{10.1111/jori.12406}
\end{barticle}
\endbibitem

\bibitem[\protect\citeauthoryear{Barigou et~al.}{2022}]{barigou2022insurance}
\begin{barticle}
\bauthor{\bsnm{Barigou}, \binits{K.}},
\bauthor{\bsnm{Bignozzi}, \binits{V.}},
\bauthor{\bsnm{Tsanakas}, \binits{A.}}:
\batitle{Insurance valuation: A two-step generalised regression approach}.
\bjtitle{ASTIN Bulletin}
\bvolume{52}(\bissue{1}),
\bfpage{211}--\blpage{245}
(\byear{2022})
\doiurl{10.1017/asb.2021.31}
\end{barticle}
\endbibitem

\bibitem[\protect\citeauthoryear{Barigou et~al.}{2019}]{barigou2019fair}
\begin{barticle}
\bauthor{\bsnm{Barigou}, \binits{K.}},
\bauthor{\bsnm{Chen}, \binits{Z.}},
\bauthor{\bsnm{Dhaene}, \binits{J.}}:
\batitle{Fair dynamic valuation of insurance liabilities: Merging actuarial
  judgement with market- and time-consistency}.
\bjtitle{Insurance: Mathematics and Economics}
\bvolume{88},
\bfpage{19}--\blpage{29}
(\byear{2019})
\doiurl{10.1016/j.insmatheco.2019.05.003}
\end{barticle}
\endbibitem

\bibitem[\protect\citeauthoryear{Bauer et~al.}{2012}]{bauer2012calculation}
\begin{barticle}
\bauthor{\bsnm{Bauer}, \binits{D.}},
\bauthor{\bsnm{Reuss}, \binits{A.}},
\bauthor{\bsnm{Singer}, \binits{D.}}:
\batitle{On the calculation of the solvency capital requirement based on nested
  simulations}.
\bjtitle{ASTIN Bulletin}
\bvolume{42}(\bissue{2}),
\bfpage{453}--\blpage{499}
(\byear{2012})
\doiurl{10.2143/AST.42.2.2182805}
\end{barticle}
\endbibitem

\bibitem[\protect\citeauthoryear{Christiansen and
  Niemeyer}{2014}]{christiansen2014fundamental}
\begin{barticle}
\bauthor{\bsnm{Christiansen}, \binits{M.C.}},
\bauthor{\bsnm{Niemeyer}, \binits{A.}}:
\batitle{Fundamental definition of the solvency capital requirement in
  {Solvency II}}.
\bjtitle{ASTIN Bulletin}
\bvolume{44}(\bissue{3}),
\bfpage{501}--\blpage{533}
(\byear{2014})
\doiurl{10.1017/asb.2014.10}
\end{barticle}
\endbibitem

\bibitem[\protect\citeauthoryear{Dhaene et~al.}{2017}]{dhaene2017fair}
\begin{barticle}
\bauthor{\bsnm{Dhaene}, \binits{J.}},
\bauthor{\bsnm{Stassen}, \binits{B.}},
\bauthor{\bsnm{Barigou}, \binits{K.}},
\bauthor{\bsnm{Linders}, \binits{D.}},
\bauthor{\bsnm{Chen}, \binits{Z.}}:
\batitle{Fair valuation of insurance liabilities: Merging actuarial judgement
  and market-consistency}.
\bjtitle{Insurance: Mathematics and Economics}
\bvolume{76},
\bfpage{14}--\blpage{27}
(\byear{2017})
\doiurl{10.1016/j.insmatheco.2017.06.003}
\end{barticle}
\endbibitem

\bibitem[\protect\citeauthoryear{Diers and Linde}{2013}]{diers2013multiyear}
\begin{barticle}
\bauthor{\bsnm{Diers}, \binits{D.}},
\bauthor{\bsnm{Linde}, \binits{M.}}:
\batitle{The multi-year non-life insurance risk in the additive loss reserving
  model}.
\bjtitle{Insurance: Mathematics and Economics}
\bvolume{52}(\bissue{3}),
\bfpage{590}--\blpage{598}
(\byear{2013})
\doiurl{10.1016/j.insmatheco.2013.03.004}
\end{barticle}
\endbibitem

\bibitem[\protect\citeauthoryear{Diers et~al.}{2016}]{diers2016addendum}
\begin{barticle}
\bauthor{\bsnm{Diers}, \binits{D.}},
\bauthor{\bsnm{Linde}, \binits{M.}},
\bauthor{\bsnm{Hahn}, \binits{L.}}:
\batitle{Addendum to ``the multi-year non-life insurance risk in the additive
  reserving model'' [{Insurance Math. Econom.} 52(3) (2013) 590--598]:
  Quantification of multi-year non-life insurance risk in chain-ladder
  reserving models}.
\bjtitle{Insurance: Mathematics and Economics}
\bvolume{67},
\bfpage{187}--\blpage{199}
(\byear{2016})
\doiurl{10.1016/j.insmatheco.2015.10.013}
\end{barticle}
\endbibitem

\bibitem[\protect\citeauthoryear{Dreksler et~al.}{2015}]{dreksler2015solvency}
\begin{barticle}
\bauthor{\bsnm{Dreksler}, \binits{S.}},
\bauthor{\bsnm{Allen}, \binits{C.}},
\bauthor{\bsnm{Akoh-Arrey}, \binits{A.}},
\bauthor{\bsnm{Courchene}, \binits{J.A.}},
\bauthor{\bsnm{Junaid}, \binits{B.}},
\bauthor{\bsnm{Kirk}, \binits{J.}},
\bauthor{\bsnm{Lowe}, \binits{W.}},
\bauthor{\bsnm{O'Dea}, \binits{S.}},
\bauthor{\bsnm{Piper}, \binits{J.}},
\bauthor{\bsnm{Shah}, \binits{M.}},
\bauthor{\bsnm{Shaw}, \binits{G.}},
\bauthor{\bsnm{Storman}, \binits{D.}},
\bauthor{\bsnm{Thaper}, \binits{S.}},
\bauthor{\bsnm{Thomas}, \binits{L.}},
\bauthor{\bsnm{Wheatley}, \binits{M.}},
\bauthor{\bsnm{Wilson}, \binits{M.}}:
\batitle{{Solvency II} technical provisions for general insurers}.
\bjtitle{British Actuarial Journal}
\bvolume{20}(\bissue{1}),
\bfpage{7}--\blpage{129}
(\byear{2015})
\doiurl{10.1017/S1357321714000099}
\end{barticle}
\endbibitem

\bibitem[\protect\citeauthoryear{England et~al.}{2019}]{england2019lifetime}
\begin{barticle}
\bauthor{\bsnm{England}, \binits{P.D.}},
\bauthor{\bsnm{Verrall}, \binits{R.J.}},
\bauthor{\bsnm{W{\"u}thrich}, \binits{M.V.}}:
\batitle{On the lifetime and one-year views of reserve risk, with application
  to {IFRS 17} and {Solvency II} risk margins}.
\bjtitle{Insurance: Mathematics and Economics}
\bvolume{85},
\bfpage{74}--\blpage{88}
(\byear{2019})
\doiurl{10.1016/j.insmatheco.2018.12.002}
\end{barticle}
\endbibitem

\bibitem[\protect\citeauthoryear{Engsner et~al.}{2017}]{engsner2017insurance}
\begin{barticle}
\bauthor{\bsnm{Engsner}, \binits{H.}},
\bauthor{\bsnm{Lindholm}, \binits{M.}},
\bauthor{\bsnm{Lindskog}, \binits{F.}}:
\batitle{Insurance valuation: A computable multi-period cost-of-capital
  approach}.
\bjtitle{Insurance: Mathematics and Economics}
\bvolume{72},
\bfpage{250}--\blpage{264}
(\byear{2017})
\doiurl{10.1016/j.insmatheco.2016.12.002}
\end{barticle}
\endbibitem

\bibitem[\protect\citeauthoryear{{European
  Commission}}{2026}]{europeancommission2026delegated}
\begin{botherref}
\oauthor{\bsnm{{European Commission}}}:
Commission Delegated Regulation ({EU}) 2026/269 of 29 October 2025 Amending
  Delegated Regulation ({EU}) 2015/35 as Regards Technical Provisions,
  Long-Term Guarantee Measures, Own Funds, Equity Risk, Spread Risk on
  Securitisation Positions, Other Standard Formula Capital Requirements,
  Reporting and Disclosure, Proportionality and Group Solvency.
Official Journal of the European Union, L Series, 18 February 2026
(2026).
\url{https://eur-lex.europa.eu/eli/reg_del/2026/269/oj}
\end{botherref}
\endbibitem

\bibitem[\protect\citeauthoryear{{European Insurance and Occupational Pensions
  Authority}}{2026a}]{eiopa2026riskmarginrts}
\begin{botherref}
\oauthor{\bsnm{{European Insurance and Occupational Pensions Authority}}}:
Final report on draft regulatory technical standards amending {Article 58} of
  {Commission Delegated Regulation ({EU}) 2015/35} on the simplified
  calculation of the risk margin.
Technical Report EIOPA-BoS-26/289,
EIOPA
(July 2026).
\url{https://www.eiopa.europa.eu/publications/final-report-rts-simplified-calculation-risk-margin-solvency-ii-review_en}
\end{botherref}
\endbibitem

\bibitem[\protect\citeauthoryear{{European Insurance and Occupational Pensions
  Authority}}{2026b}]{eiopa2026valuation}
\begin{botherref}
\oauthor{\bsnm{{European Insurance and Occupational Pensions Authority}}}:
Final report on revised {Guidelines} on valuation of technical provisions.
Technical Report EIOPA-BoS-26/285,
EIOPA
(July 2026).
\url{https://www.eiopa.europa.eu/publications/final-report-revised-guidelines-valuation-technical-provisions-solvency-ii-review_en}
\end{botherref}
\endbibitem

\bibitem[\protect\citeauthoryear{{European Parliament and Council of the
  European Union}}{2025}]{europeanunion2025solvency}
\begin{botherref}
\oauthor{\bsnm{{European Parliament and Council of the European Union}}}:
Directive ({EU}) 2025/2 of the European Parliament and of the Council of 27
  November 2024 Amending Directive 2009/138/EC as Regards Proportionality,
  Quality of Supervision, Reporting, Long-Term Guarantee Measures,
  Macro-Prudential Tools, Sustainability Risks and Group and Cross-Border
  Supervision, and Amending Directives 2002/87/EC and 2013/34/EU.
Official Journal of the European Union, L Series, 8 January 2025
(2025).
\url{https://eur-lex.europa.eu/eli/dir/2025/2/oj}
\end{botherref}
\endbibitem

\bibitem[\protect\citeauthoryear{Floreani}{2011}]{floreani2011risk}
\begin{barticle}
\bauthor{\bsnm{Floreani}, \binits{A.}}:
\batitle{{Risk Margin} estimation through the cost of capital approach: Some
  conceptual issues}.
\bjtitle{The Geneva Papers on Risk and Insurance - Issues and Practice}
\bvolume{36}(\bissue{2}),
\bfpage{226}--\blpage{253}
(\byear{2011})
\doiurl{10.1057/gpp.2011.2}
\end{barticle}
\endbibitem

\bibitem[\protect\citeauthoryear{Gambaro}{2025}]{gambaro2025capital}
\begin{barticle}
\bauthor{\bsnm{Gambaro}, \binits{A.M.}}:
\batitle{The capital-on-capital cost in {Solvency II} {Risk Margin}}.
\bjtitle{European Actuarial Journal}
\bvolume{15}(\bissue{1}),
\bfpage{137}--\blpage{161}
(\byear{2025})
\doiurl{10.1007/s13385-024-00401-8}
\end{barticle}
\endbibitem

\bibitem[\protect\citeauthoryear{Huang et~al.}{1977}]{huang1977bounds}
\begin{barticle}
\bauthor{\bsnm{Huang}, \binits{C.C.}},
\bauthor{\bsnm{Ziemba}, \binits{W.T.}},
\bauthor{\bsnm{Ben-Tal}, \binits{A.}}:
\batitle{Bounds on the expectation of a convex function of a random variable:
  With applications to stochastic programming}.
\bjtitle{Operations Research}
\bvolume{25}(\bissue{2}),
\bfpage{315}--\blpage{325}
(\byear{1977})
\doiurl{10.1287/opre.25.2.315}
\end{barticle}
\endbibitem

\bibitem[\protect\citeauthoryear{M{\"o}hr}{2011}]{mohr2011market}
\begin{barticle}
\bauthor{\bsnm{M{\"o}hr}, \binits{C.}}:
\batitle{Market-consistent valuation of insurance liabilities by cost of
  capital}.
\bjtitle{{ASTIN} Bulletin}
\bvolume{41}(\bissue{2}),
\bfpage{315}--\blpage{341}
(\byear{2011})
\doiurl{10.2143/AST.41.2.2136980}
\end{barticle}
\endbibitem

\bibitem[\protect\citeauthoryear{Ohlsson and
  Lauzeningks}{2009}]{ohlsson2009oneyear}
\begin{barticle}
\bauthor{\bsnm{Ohlsson}, \binits{E.}},
\bauthor{\bsnm{Lauzeningks}, \binits{J.}}:
\batitle{The one-year non-life insurance risk}.
\bjtitle{Insurance: Mathematics and Economics}
\bvolume{45}(\bissue{2}),
\bfpage{203}--\blpage{208}
(\byear{2009})
\doiurl{10.1016/j.insmatheco.2009.06.001}
\end{barticle}
\endbibitem

\bibitem[\protect\citeauthoryear{Pelkiewicz
  et~al.}{2020}]{pelkiewicz2020review}
\begin{barticle}
\bauthor{\bsnm{Pelkiewicz}, \binits{A.J.}},
\bauthor{\bsnm{Ahmed}, \binits{S.W.}},
\bauthor{\bsnm{Fulcher}, \binits{P.}},
\bauthor{\bsnm{Johnson}, \binits{K.L.}},
\bauthor{\bsnm{Reynolds}, \binits{S.M.}},
\bauthor{\bsnm{Schneider}, \binits{R.J.}},
\bauthor{\bsnm{Scott}, \binits{A.J.}}:
\batitle{A review of the {Risk Margin}: {Solvency II} and beyond}.
\bjtitle{British Actuarial Journal}
\bvolume{25},
\bfpage{1}
(\byear{2020})
\doiurl{10.1017/S135732172000001X}
\end{barticle}
\endbibitem

\bibitem[\protect\citeauthoryear{Pelsser and Stadje}{2014}]{pelsser2014time}
\begin{barticle}
\bauthor{\bsnm{Pelsser}, \binits{A.}},
\bauthor{\bsnm{Stadje}, \binits{M.}}:
\batitle{Time-consistent and market-consistent evaluations}.
\bjtitle{Mathematical Finance}
\bvolume{24}(\bissue{1}),
\bfpage{25}--\blpage{65}
(\byear{2014})
\doiurl{10.1111/mafi.12026}
\end{barticle}
\endbibitem

\bibitem[\protect\citeauthoryear{Salzmann and
  W{\"u}thrich}{2010}]{salzmann2010cost}
\begin{barticle}
\bauthor{\bsnm{Salzmann}, \binits{R.}},
\bauthor{\bsnm{W{\"u}thrich}, \binits{M.V.}}:
\batitle{Cost-of-capital margin for a general insurance liability runoff}.
\bjtitle{ASTIN Bulletin}
\bvolume{40}(\bissue{2}),
\bfpage{415}--\blpage{451}
(\byear{2010})
\doiurl{10.2143/AST.40.2.2061123}
\end{barticle}
\endbibitem

\bibitem[\protect\citeauthoryear{Shaked and
  Shanthikumar}{2007}]{shaked2007stochastic}
\begin{bbook}
\bauthor{\bsnm{Shaked}, \binits{M.}},
\bauthor{\bsnm{Shanthikumar}, \binits{J.G.}}:
\bbtitle{Stochastic Orders}.
\bsertitle{Springer Series in Statistics}.
\bpublisher{Springer},
\blocation{New York}
(\byear{2007}).
\doiurl{10.1007/978-0-387-34675-5}
\end{bbook}
\endbibitem

\bibitem[\protect\citeauthoryear{Tsanakas et~al.}{2013}]{tsanakas2013market}
\begin{barticle}
\bauthor{\bsnm{Tsanakas}, \binits{A.}},
\bauthor{\bsnm{W{\"u}thrich}, \binits{M.V.}},
\bauthor{\bsnm{{\v{C}}ern{\'y}}, \binits{A.}}:
\batitle{Market value margin via mean--variance hedging}.
\bjtitle{ASTIN Bulletin}
\bvolume{43}(\bissue{3}),
\bfpage{301}--\blpage{322}
(\byear{2013})
\doiurl{10.1017/asb.2013.18}
\end{barticle}
\endbibitem

\bibitem[\protect\citeauthoryear{Waszink}{2024}]{waszink2024comments}
\begin{barticle}
\bauthor{\bsnm{Waszink}, \binits{H.}}:
\batitle{Comments on the {Solvency II} {Risk Margin} and proposed amendments}.
\bjtitle{British Actuarial Journal}
\bvolume{29},
\bfpage{7}
(\byear{2024})
\doiurl{10.1017/S1357321724000047}
\end{barticle}
\endbibitem

\bibitem[\protect\citeauthoryear{W{\"u}thrich et~al.}{2011}]{wuthrich2011risk}
\begin{barticle}
\bauthor{\bsnm{W{\"u}thrich}, \binits{M.V.}},
\bauthor{\bsnm{Embrechts}, \binits{P.}},
\bauthor{\bsnm{Tsanakas}, \binits{A.}}:
\batitle{Risk margin for a non-life insurance run-off}.
\bjtitle{Statistics \& Risk Modeling}
\bvolume{28}(\bissue{4}),
\bfpage{299}--\blpage{317}
(\byear{2011})
\doiurl{10.1524/strm.2011.1096}
\end{barticle}
\endbibitem

\bibitem[\protect\citeauthoryear{W{\"u}thrich
  et~al.}{2009}]{wuthrich2009uncertainty}
\begin{barticle}
\bauthor{\bsnm{W{\"u}thrich}, \binits{M.V.}},
\bauthor{\bsnm{Merz}, \binits{M.}},
\bauthor{\bsnm{Lysenko}, \binits{N.}}:
\batitle{Uncertainty of the claims development result in the chain ladder
  method}.
\bjtitle{Scandinavian Actuarial Journal}
\bvolume{2009}(\bissue{1}),
\bfpage{63}--\blpage{84}
(\byear{2009})
\doiurl{10.1080/03461230801979732}
\end{barticle}
\endbibitem

\end{thebibliography}

\end{document}